\documentclass[a4paper,UKenglish,numberwithinsect]{article}

\usepackage[left=2.5cm,right=2.5cm,top=2.5cm,bottom=3cm]{geometry}
\usepackage{tikz}
\usetikzlibrary{quotes,graphs}
\usepackage{amsmath,amsthm,amssymb}
\usepackage[pdfborder={0 0 0}]{hyperref}

\usepackage{colortbl}
\usepackage{booktabs}
\usepackage{tikz}
\usetikzlibrary{arrows.meta,graphs}

\title{On the Descriptive~Complexity of Color Coding} 
\author{Max Bannach \and Till Tantau}
\date{%
  Institute for Theoretical Computer Science,\\
  Universit\"at zu L\"ubeck\\
  L\"ubeck, Germany \\
  \texttt{\{bannach,tantau\}@tcs.uni-luebeck.de}
}

\newenvironment{notation}[3][Strong-qr:]%
{%
  \begin{notation*}[{\normalfont#2}]
    \hfill \textbf{{#1} #3}
    \begin{list}%
    {}%
    {%
      \def\labelstyle{\bfseries}
      \setlength{\topsep}{0pt}%
      \settowidth{\labelwidth}{\labelstyle Semantics}%
      \setlength{\leftmargin}{\labelwidth}%
      \addtolength{\leftmargin}{\labelsep}%
      \setlength{\itemsep}{0pt}%
      \setlength{\parsep}{0pt}%
    }%
  }%
  {%
  \end{list}
  \end{notation*}%
}

\newcommand\Class[1]{%
  \mathchoice%
  {\text{\normalfont\small$\mathrm{#1}$}}%
  {\text{\normalfont\small$\mathrm{#1}$}}%
  {\text{\normalfont$\mathrm{#1}$}}%
  {\text{\normalfont$\mathrm{#1}$}}%
}

\newcommand\Para{\mathrm{para\text-}}

\newcommand{\struc}{\Lang{struc}}
\newcommand{\Lang}[1]{\text{\normalfont\textsc{#1}}}
\newcommand{\PLang}[2][]{\mathrm{p}_{#1}\Lang{-#2}}

\newcommand{\existseqk}[1]{\exists^{|.|=k}\kern-1pt#1\,}
\newcommand{\existslek}[1]{\exists^{|.|\le k}\kern-1pt#1\,}

\newcommand{\op}[1]{\operatorname{\textsc{#1}}}

\newtheorem{theorem}{Theorem}[section]

\newtheorem{lemma}[theorem]{Lemma}

\theoremstyle{plain}
\newtheorem{example}[theorem]{Example}

\newtheorem{definition}[theorem]{Definition}
\newtheorem{fact}[theorem]{Fact}
\newtheorem*{claim*}{Claim}
\newtheorem*{corollary*}{Corollary}

\theoremstyle{definition}
\newtheorem*{notation*}{Notation}

\tikzset{graph/.style = {
    node/.style = {circle, minimum size=5mm, inner sep=1pt, semithick, draw, font=\small},
    thin,
    > = {Stealth[round,sep]}
  }
}

\begin{document}

\maketitle

\begin{abstract}
  Color coding is an algorithmic technique used in parameterized
  complexity theory to detect ``small'' structures inside graphs. The
  idea is to derandomize algorithms that first randomly
  color a graph and then search for an easily-detectable, small color
  pattern. We transfer color coding to the world of descriptive
  complexity theory by characterizing -- purely in terms of the
  syntactic structure of describing formulas -- when the powerful
  second-order quantifiers representing a random coloring can be
  replaced by equivalent, simple first-order formulas. Building on
  this result, we identify syntactic properties of first-order
  quantifiers that can be eliminated from formulas describing
  parameterized problems. The result applies to many packing and
  embedding problems, but also to the long path problem. Together with a
  new result on the parameterized complexity of formula families
  involving only a fixed number of variables, we get that many
  problems lie in \textsc{fpt} just because of the way they are
  commonly described using logical formulas.     
\end{abstract}


\allowdisplaybreaks

\section{Introduction}

Descriptive complexity provides a powerful link between logic
and complexity theory: We use a logical formula to
\emph{describe} a problem and can then infer the computational
complexity of the problem just from the \emph{syntactic structure} of
the formula. As a striking example,  Fagin's Theorem~\cite{Fagin1974}
tells us that 3-colorability lies in $\Class{NP}$ just because its
describing formula (``there exist three colors such that all adjacent
vertex pairs have different colors'')  
is an existential second-order formula.
In the context of fixed-parameter tractability theory, methods from descriptive
complexity are also used a lot -- but commonly to show that problems are 
\emph{difficult.} For instance, the A- and W-hierarchies are
defined in logical terms~\cite{FlumG2006}, but their hard problems are
presumably ``beyond'' the class $\Class{FPT}$ of
fixed-parameter tractable problems.

The methods of descriptive complexity are only rarely used to show
that problems are \emph{in} $\Class{FPT}$. 
More precisely, the syntactic structure of the natural logical
descriptions of standard parameterized problems found in textbooks are
not known to imply that the problems lie in $\Class{FPT}$ -- even though this
is known to be the case for many of them.
To appreciate the underlying difficulties, consider the following three
parameterized problems: $\PLang{matching}$,
$\PLang{triangle-packing}$, and $\PLang{clique}$. In each case, we are
given an undirected graph as input and a number $k$ and we are then
asked whether the graph contains $k$ vertex-disjoint edges (a size-$k$
matching), $k$~vertex-disjoint triangles, or a clique of size~$k$,
respectively. The problems 
are known to have widely different complexities (maximal matchings can
actually be found in polynomial time, triangle packing lies at least
in $\Class{FPT}$, while finding cliques is $\Class W[1]$-complete) but
\emph{very} similar logical descriptions: 
\begin{align}
  \alpha_k &= \exists x_1 \cdots \exists x_{2k} \bigl( \textstyle
              \bigwedge_{i\neq j} x_i \neq x_j \land \bigwedge_{i=1}^k
             E x_{2i-1} x_{2i} \bigr), \label{eq:alpha}\\
  \beta_k &= \exists x_1 \cdots \exists x_{3k} \bigl( \textstyle
              \bigwedge_{i\neq j} x_i \neq x_j \land \bigwedge_{i=1}^k
             (E x_{3i-2} x_{3i-1} \land E x_{3i-2} x_{3i} \land E x_{3i-1} x_{3i}) \bigr), \label{eq:beta}\\
  \gamma_k &= \exists x_1 \cdots \exists x_{k} \bigl( \textstyle
             \bigwedge_{i\neq j} x_i \neq x_j \land \bigwedge_{i \neq
             j} E x_i x_j \bigr). \label{eq:gamma}
\end{align}
The family $(\alpha_k)_{k\in\mathbb N}$ of formulas is clearly a natural
``slicewise'' description of the matching problem: A graph $\mathcal
G$ has a size-$k$ matching if, and only if, $\mathcal G \models
\alpha_k$. The families $(\beta_k)_{k\in\mathbb N}$ and $(\gamma_k)_{k\in\mathbb N}$ are natural
parameterized descriptions of the triangle packing and the clique
problems, respectively. 
Well-known results on the descriptive complexity of parameterized
problems allow us to infer~\cite{FlumG2006} from the above descriptions that all
three problems lie in $\Class W[1]$, but offer no hint why the first
two problems actually lie in the class $\Class{FPT}$ -- syntactically
the clique problem arguably ``looks like the easiest 
one'' when in fact it is semantically the most difficult one. The results of
this paper will remedy this: We will show that the syntactic
structures of the formulas $\alpha_k$ and $\beta_k$ imply membership
of $\PLang{matching}$ and $\PLang{triangle-packing}$ in $\Class{FPT}$.

The road to deriving the computational complexity of parameterized
problems just from the syntactic properties of slicewise first-order
descriptions involves three major steps: First, a 
characterization of when the color coding technique is applicable in
terms of syntactic properties of second-order quantifiers. Second, an
exploration of how these results on second-order formulas apply to
first-order formulas, leading to the notion of \emph{strong} and
\emph{weak} quantifiers and to an elimination theorem for weak
quantifiers. Third, we add a new characterization to the body of known
characterizations of how classes like $\Class{FPT}$ can be characterized
in a slicewise fashion by logical formulas.

\paragraph*{Our Contributions I: A Syntactic Characterization of Color Coding.}

The hard triangle packing problem from above becomes almost trivial when
we just wish to check whether a \emph{vertex-colored} graph contains a
red triangle, a green triangle, a blue triangle, a yellow triangle,
and so on for $k$ different colors. The ingenious idea behind the color coding 
technique of Alon, Yuster, and Zwick~\cite{AlonYZ95} is to reduce the
original problem to the much simpler colored version by simply 
\emph{randomly coloring the graph.} Of course, even if there are
$k$~disjoint triangles, we will most likely \emph{not} color them
monochromatically and differently, \emph{but} the probability of 
``getting lucky'' is nonzero and depends only on the
parameter~$k$. Even better, Alon et~al.\ point out that one can 
\emph{derandomize the coloring easily by using universal hash
  functions to color each vertex with its hash value.}

Applying this idea in the setting of descriptive
complexity was recently pioneered by Chen et~al.~\cite{ChenFH2017}. Transferred to the triangle
packing problem, their argument  would roughly be: ``Testing  for
each color~$i$ whether there is a monochromatic triangle of color~$i$ can be
done in first-order logic using something like  
$\bigwedge_{i=1}^k \exists x \exists y \exists z 
  (E x y \land E y z \land E x z \land C_i x \land C_i y \land C_i z)$.
Next,
instead of testing  whether $x$ has color~$i$ using the formula $C_i
x$, we can test whether $x$ gets hashed to $i$ by a hash
function. Finally, since computing appropriate universal hash 
functions only involves addition and multiplication, we can express
the derandomized algorithm using an arithmetic first-order formula of low
quantifier rank.'' Phrased differently, Chen et~al.\ would argue that
$\bigwedge_{i=1}^k \exists x \exists y \exists z 
(E x y \land E y z \land E x z \land C_i x \land C_i y \land C_i z)$
together with the requirement that the $C_i$ are pairwise disjoint
is (ignoring some details) equivalent to $
\delta_k = \exists p \exists q \textstyle
\bigwedge_{i=1}^k \exists x \exists y 
\exists z (E x y \land E y z \land E x z \land \op{hash}_{k}(x,p,q)=i \land
\op{hash}_{k}(y,p,q)=i \land \op{hash}_{k}(z,p,q)=i)$,
where $\op{hash}_{k}(x,p,q)=i$ is a formula that is true when ``$x$ is
hashed to~$i$ by a member of a universal family of hash functions
indexed by $q$ and~$p$.'' 

The family $(\delta_k)_{k\in\mathbb N}$ may seem rather technical and, indeed, its
importance becomes visible only in conjunction with another result by
Chen et al.~\cite{ChenFH2017}: They show that a parameterized problem
lies in $\Para\Class{AC}^0$, one of the smallest ``sensible''
subclasses of $\Class{FPT}$, if it can be described by a family
$(\phi_k)_{k \in \mathbb N}$ of $\Class{FO}[+,\times]$ formulas of
\emph{bounded quantifier rank} such that the finite models of $\phi_k$
are exactly the elements of the $k$th slice of the problem. Since the
triangle packing problem can be described in this way via the family
$(\delta_k)_{k\in\mathbb N}$ of formulas, all of which have a
quantifier rank $5$ plus the constant number of quantifiers used to
express the arithmetics in the formulas $\op{hash}_k(x,p,q)=i$, we get
$\PLang{triangle-packing} \in \Class{FPT}$. 

Clearly, this beautiful idea cannot work in all situations: If it
also worked for the formula mentioned earlier expressing
3-colorability, 3-colorability would be first-order expressible, which
is known to be impossible. 
Our first main contribution is a \emph{syntactic characterization of
  when the color coding technique is applicable,} that is, of why
color coding works for triangle packing but not for 3-colorability:
For triangle packing, the colors $C_i$ are applied to variables only 
inside \emph{existential scopes} (``$\exists x \exists y \exists z$'')
while for 3-colorability the colors $R$, $G$, and $B$ are also applied
to variables inside universal scopes (``for all adjacent
vertices''). In general, see Theorem~\ref{thm-cc-intro} for the
details, we show that a second-order quantification over an arbitrary number of
disjoint colors $C_i$ can be replaced by a fixed number of first-order 
quantifiers whenever none of the $C_i$ is used in a universal scope.

\paragraph*{Our Contributions II: New First-Order Quantifier Elimination Rules.}

The ``purpose'' of the colors~$C_i$ in the formulas $\bigwedge_{i=1}^k
\exists x \exists y \exists z (E x y \land E y z \land E x z \land C_i x
\land C_i y \land C_i z)$ is not that the three vertices of a
triangle get a particular color, but just that they get a color
\emph{different} from the color of all other triangles. Indeed, our
``real'' objective in these formulas is to ensure that the vertices of
a triangle are \emph{distinct} from the vertices in the other
triangles -- and giving vertices different colors is ``just a means''
of ensuring this. 

In our second main contribution we explore this idea further: If the
main (indeed, the only) use of colors in the context of color
coding is to ensure that certain vertices are different, let us do
away with colors and instead focus on the notion of
\emph{distinctness.} To better explain this idea, consider the
following family, also describing triangle packing, where the only change is that we now require 
(a bit superfluously) that even the vertices inside a triangle get
different colors: $\bigwedge_{j=1}^k \exists x \exists y \exists z (E
x y \land E y z \land E x z \land C_{3j-2} x \land C_{3j-1} y \land C_{3j}
z)$. 
Observe that each $C_i$ is now applied to exactly one variable ($x$,
$y$, or $z$ in one of the many literals) and the only ``effect'' that
all these applications have is to ensure that the variables are
different. In particular, the formula is equivalent to
\begin{align}
  \exists x_1 \cdots \exists x_{3k}\textstyle \bigwedge_{i\neq j} x_i \neq x_j
  \land \bigwedge_{j=1}^k \exists x \exists y \exists z (&E x y \land
                                                           E y z
                                                           \land E x z \land{}\notag\\[-2mm]
  &x_{3j-2} = x \land x_{3j-1} =y \land x_{3j} =z) \label{eq:tri}
\end{align}
and these formulas are clearly equivalent to the almost identical
formulas from~\eqref{eq:beta}.

In a sense, in \eqref{eq:tri} the many existential quantifiers
$\exists x_i$ and the many $x_i \neq x_j$ literals come ``for free'' from
the color coding technique, while $\exists x$, $\exists y$, and
$\exists z$ have nothing to do with color coding.  
Our key observation is a syntactic property that tells us
whether a quantifier comes ``for free'' in this way (we will call it
\emph{weak}) or not (we will call it \emph{strong}):
Definition~\ref{def-weak} states (essentially) that weak quantifiers have the form
$\exists x (\phi)$ such that $x$ is not free in any universal scope of
$\phi$ and $x$ is used in at most one literal that is not of the form
$x \neq y$. To make weak quantifiers easier to 
spot, we mark their bound variables with a dot (note that this is a
``syntactic hint'' without semantic meaning). Formulas~\eqref{eq:tri}
now read $\textstyle\exists \dot x_1 \cdots\exists \dot x_{3k}
\penalty100 \bigwedge_{i\neq j} \dot x_i \neq \dot x_j \land\penalty0
\bigwedge_{j=1}^k \exists x \exists y \exists z (E x y  \land E x z  \land
E y z \land \penalty0  \dot x_{3j-2} = x \land \penalty0 \dot x_{3j-1}
=y \land \penalty0  \dot x_{3j} =z)$. Observe that $x$, $y$, and $z$ are
not weak since each is used in three literals that are not
inequalities. 

We show in Theorem~\ref{theorem-strong-character} that each $\phi$
is equivalent to a $\phi'$ whose quantifier rank depends only on the
\emph{strong quantifier rank} 
of~$\phi$ (meaning that we ignore the weak quantifiers) and whose
number of variables depends only on the number of strong
variables in~$\phi'$. For instance, the formulas from \eqref{eq:tri}
all have strong quantifier rank~$3$ and, thus, the triangle packing
problem can be described by a family of constant (normal) quantifier 
rank. Applying Chen et al.'s characterization yields membership
in~$\Para\Class{AC}^0$.

As a more complex example, let us sketch a ``purely syntactic'' proof of the
result \cite{BannachT2016,ChenM2015} that the embedding problem for
graphs~$H$ of tree depth at most~$d$ lies in $\Para\Class{AC}^0$ for
each~$d$. Once more, we construct a family $(\phi_H)$ of formulas of
constant strong quantifier rank that describes the problem. For a
graph $H$ and a rooted tree~$T$ of depth~$d$ such that $H$ is contained in
$T$'s transitive closure (this is the definition of ``$H$ has tree
depth~$d$''), let $c_1$ be the root of~$T$ and let
$\operatorname{children}(c)$ be the children of~$c$ in~$T$. Then the
following formula of strong quantifier rank~$d$ describes that $H$ can
be embedded into a structure:
\begin{align*}
  \exists \dot x_1 \cdots \exists \dot x_{|H|} \bigl(
  \textstyle\bigwedge_{i\neq j} \dot x_i \neq \dot x_j \land{}
  &\exists n_1 (n_1 = \dot x_{c_1} \land 
    \textstyle\bigwedge_{c_2 \in \operatorname{children}(c_1)}\exists n_2 (n_2 = \dot x_{c_2} \land {}\\
  \textstyle\bigwedge_{c_3 \in \operatorname{children}(c_2)}
  &\exists n_3 (n_3 = \dot x_{c_3} \land
    \textstyle\bigwedge_{c_4 \in \operatorname{children}(c_3)}\exists n_4 (n_4 = \dot x_{c_4} \land  \dots \\
  \textstyle\bigwedge_{c_d \in \operatorname{children}(c_{d-1})}
  &\exists n_d (n_d = \dot x_{c_d} \land \textstyle\bigwedge_{i,j \in
    \{1,\dots,d\}:(c_i,c_j) \in E(H)} En_in_j)\dots)))\bigr).
\end{align*}

\paragraph*{Our Contributions III: Slicewise Descriptions and Variable Set Sizes.}

Our third contribution is a new result in the same vein as the already
repeatedly mentioned result of Chen et al.~\cite{ChenFH2017}:
Theorem~\ref{thm-ac0up} states that a parameterized problem can be
described slicewise by a family 
$(\phi_k)_{k\in\mathbb  N}$ of arithmetic first-order formulas that
all use only a \emph{bounded number of variables} if, and only if, the
problem lies 
in $\Para\Class{AC}^{0\uparrow}$ -- a class that has been encountered
repeatedly in the literature~\cite{BannachST2015,BannachT2016,DasER2018,PilipczukST18}, but for which no
characterization was known. It contains all parameterized problems that
can be decided by $\Class{AC}$-circuits whose depth depends only on the parameter
and whose size is of the form $f(k)\cdot n^{c}$.

As an example, consider the problem of deciding whether a
graph contains a path of length~$k$  
(no vertex may be visited twice). It can be described (for odd~$k$) by: $
  \exists s \exists t \exists x ( Esx \land \exists \dot x_1 (\dot x_1 = x \land
                      \exists y ( Exy \land \exists \dot x_2 (\dot x_2 = y \land
                      \exists x ( Eyx \land \exists \dot x_3 (\dot x_3 = x \land
                      \exists y ( Exy \land \exists \dot x_4
                      (\penalty10 \dot x_4 = y \land
                      \cdots \land
                      \exists x (Eyx \land
                        x=t \land \exists \dot x_k (\dot x_k = x \land  { \bigwedge_{i\neq j} \dot x_i
                        \neq \dot x_j})\dots))))$.
Note that, now, the strong quantifier rank depends on~$k$ and, thus,
is not constant. However, there are now only four strong variables,
namely $s$, $t$, $x$, and $y$. By
Theorem~\ref{theorem-strong-character} we see that the 
above formulas are equivalent to a family of formulas with a bounded
number of variables and by Theorem~\ref{thm-ac0up} we see that $\PLang{long-path}
\in \Para\Class{AC}^{0\uparrow} \subseteq \Class{FPT}$.
These ideas also generalize easily and we give a purely syntactic
proof of the seminal result from the original color coding
paper~\cite{AlonYZ95} that the embedding problem for graphs of bounded
tree \emph{width} lies in~$\Class{FPT}$. The core observation -- which
unifies the results for tree width and depth -- is 
that for each graph with a given tree decomposition, the embedding
problem can be described by a formula whose strong nesting structure
mirrors the tree structure and whose strong variables mirror the bag
contents.

\paragraph*{Related Work.}

Flum and Grohe~\cite{FlumG2003} were the first to give
characterizations of $\Class{FPT}$ and of many subclasses in terms of
the syntactic properties of formulas describing
their members. Unfortunately, these syntactic properties do not hold for the
descriptions of parameterized problems found in the literature. For
instance, they show that $\Class{FPT}$ contains exactly the problems 
that can be described by families of $\Class{FO}[\op{lfp}]$-formulas
of bounded quantifier rank -- but actually describing problems like
$\PLang{vertex-cover}$ in this way is more or less hopeless and yields
little insights into the structure or complexity of the problem. We
believe that it is no coincidence that no applications  of these beautiful characterizations to concrete
problems could be found in the
literature -- at least prior to very recent work by Chen and
Flum~\cite{ChenF2018}, who study slicewise descriptions of problems on
structures of bounded tree depth, and the already cited
article of Chen et al.~\cite{ChenFH2017}, who \emph{do}
present a family of formulas that describe the vertex cover
problem. This family internally uses the color
coding technique and is thus closely related to our
results. The crucial difference is, however, that we identify
syntactic properties of logical formulas that imply that the color
coding technique can be applied. It then suffices to find a family
describing a given problem that meets the syntactic properties to
establish the complexity of the problem: there is no need to actually
construct the color-coding-based formulas -- indeed, there is not even
a need to understand how color coding works in order to decide whether
a quantifier is weak or strong.

\paragraph*{Organization of this Paper.}

In Section~\ref{section-describing} we first review some of the
existing work on the descriptive complexity of parameterized
problems. We add to this work in the form of the mentioned
characterization of the class $\Para\Class{AC}^{0\uparrow}$ in terms
of a bounded number of variables. Our main technical results are then
proved in Section~\ref{section-cc}, where 
we establish and prove the syntactic properties that formulas must
have in order for the color coding method to be applicable. In
Section~\ref{section-applications} we then apply the findings and show
how membership of different natural problems in $\Para\Class{AC}^0$
and $\Para\Class{AC}^{0\uparrow}$ (and, thus, in $\Class{FPT}$) can be derived
entirely from the syntactic structure of the formulas describing them.

\section{Describing Parameterized Problems}\label{section-describing}

A happy marriage of parameterized complexity and descriptive
complexity was first presented in~\cite{FlumG2003}. We first review
the most important definitions from~\cite{FlumG2003} and then prove a
new characterization, namely of the class
$\Para\Class{AC}^{0\uparrow}$ that contains all problems 
decidable by AC-circuits of parameter-dependent depth and
``$\Class{FPT}$-like'' size. Since the results and notions will be
useful later, but do not lie at the paper's heart, we
keep this section brief.

\paragraph*{Logical Terminology.}
We only consider first-order logic and use standard notations, with
the perhaps only deviations being that we write relational atoms briefly as $Exy$ 
instead of $E(x,y)$ and that the literal $x \neq y$ is an abbreviation
for $\neg\, x=y$ (recall that a \emph{literal} is an atom or a negated atom). 
Signatures, typically denoted~$\tau$, are always finite and may only contain relation 
symbols and constant symbols -- with one exception: The special unary
function symbol $\op{succ}$ may also be present in a
signature. Let us write $\op{succ}^k$ for the $k$-fold application of
$\op{succ}$, so $\op{succ}^3(x)$ is short for
$\op{succ}(\op{succ}(\penalty0\op{succ}(x)))$. It allows us
to specify any fixed non-negative integer without having to use
additional variables. An alternative is to dynamically add constant
symbols for numbers to signatures as done in~\cite{ChenFH2017}, but
we believe that following~\cite{FlumG2003} and adding the successor
function gives a leaner formal framework.
Let $\operatorname{arity}(\tau)$ be the maximum arity of relation
symbols in~$\tau$.

We denote by $\struc[\tau]$ the class of all
$\tau$-structures and by $\left|\mathcal A\right|$ the
universe of~$\mathcal A$.
As is often the case in descriptive complexity theory, we only consider
ordered structures in which the ternary predicates $\op{add}$ and
$\op{mult}$ are available and have their natural meaning. Formally, we say
$\tau$ is \emph{arithmetic} if it contains 
all of the predicates $<$, $\op{add}$, $\op{mult}$, the function
symbol $\op{succ}$, and the constant symbol~$0$ (it is included for
convenience only). In this case, $\struc[\tau]$ contains only those
$\mathcal A$ for which $<^{\mathcal A}$ is a linear ordering of
$|\mathcal A|$ and the other operations have their
natural meaning relative to $<^{\mathcal A}$ (with the successor of
the maximum element of the universe being itself and with $0$ being
the minimum with respect to~$<^{\mathcal A}$). We write $\phi
\in \Class{FO}[+,\times]$ when $\phi$ is a $\tau$-formula for an arithmetic~$\tau$.

A \emph{$\tau$-problem} is a set $Q \subseteq \struc[\tau]$
closed under isomorphisms. A $\tau$-formula $\phi$ \emph{describes} a
$\tau$-problem $Q$ if $Q = \{\mathcal A \in
\struc[\tau] \mid \mathcal A \models \phi\}$ and it \emph{describes
  $Q$ eventually} if $\phi$ describes a set $Q'$ that 
differs from $Q$ only on structures of a certain maximum size.  

\begin{lemma}\label{lemma-eventually}
  For each $\phi \in \Class{FO}[+,\times]$ that describes a $\tau$-problem $Q$ eventually,
  there are quantifier-free formulas $\alpha$ and $\beta$ such that
  $(\alpha \land \phi) \lor \beta$ describes~$Q$. 
\end{lemma}
\begin{proof}
  The statement of the lemma would be quite simple if we did not
  require $\alpha$ and~$\beta$ to be quantifier-free: Without this
  requirement, all we need to do is to use $\alpha$ and $\beta$ to
  fix $\phi$ on the finitely many (up to isomorphisms) structures
  on which $\phi$ errs by ``hard-wiring'' which of these structures
  are elements of $Q$ and which are not. However, the natural way to
  do this ``hard-wiring'' of size-$m$ structures is to use
  $m$~quantifiers to bind all elements of the universe. This is
  exactly what we do \emph{not} wish to do. Rather, we use
  the successor function to refer to the elements of the
  universe without using any quantifiers. 

  In detail, let $m$ be a number such that for all $\mathcal A \in
  \struc[\tau]$ with $\|\mathcal A\| \ge m$ (that is, the size
  $\|\mathcal A\|$ of the universe $|\mathcal A|$ is at least~$m$) we
  have $\mathcal A 
  \models \phi$ if, and only if, $\mathcal A \in Q$. We set $\alpha$
  to $\op{universe}^{\ge m}$, a shorthand for $\op{succ}^{m-1} (0)
  \neq \op{succ}^m (0)$, which is true only for universes of size at
  least~$m$.  
  We define $\beta$ so that it is true exactly for all $\tau$-structures
  $\mathcal A \in Q$ of size at most~$m$ (for simplicity we assume
  that $E^2$ is the only relation symbol in~$\tau$):  
  \begin{align*}
    \textstyle\beta = \bigwedge_{s=1}^m 
    \smash{\Bigl(}&(\op{universe}^{\ge s} \land \neg\op{universe}^{\ge s+1})  \\
    &\to\textstyle                \bigvee_{\mathcal A \in Q, |\mathcal A|=\{0,\dots,s-1\}}\Bigl( 
                 \bigwedge_{u,v \in |\mathcal A|:(u,v) \in E^{\mathcal S}} E(\op{succ}^{u}(0), \op{succ}^{v}(0))
                 \land {}\\
                 &\textstyle \phantom{\to\textstyle   \bigvee_{\mathcal A \in Q, |\mathcal
                   A|=\{0,\dots,s-1\}}\smash{\Bigl(}}
                   \bigwedge_{u,v \in |\mathcal A|:(u,v) \notin E^{\mathcal S}} \neg E(\op{succ}^{u}(0), \op{succ}^{v}(0))
                 \smash{\Bigr)\Bigr)}.\qedhere
  \end{align*}
\end{proof}

We write $\operatorname{qr}(\phi)$ for the quantifier rank of a
formula and $\operatorname{bound}(\phi)$ for the set of its bound
variables. For instance, for $\phi = \bigl(\exists x \exists y
(Exz)\bigr) \lor \forall y (Px)$ we have $\operatorname{qr}(\phi) = 2$,
since the maximum nesting is caused by the two nested existential quantifiers, and  
$\operatorname{bound}(\phi) = \{x,y\}$.

Let us say that $\phi$ is \emph{in negation normal form} if negations
are applied only to atomic formulas.

\paragraph*{Describing Parameterized Problems.}

When switching from classical complexity theory to descriptive
complexity theory, the basic change is that ``words'' get replaced by
``finite structures.'' The same idea works for parameterized
complexity theory and, following Flum and Grohe~\cite{FlumG2003}, let
us define \emph{parameterized problems} as subsets $Q \subseteq
\struc[\tau] \times \mathbb N$ where $Q$ is closed under
isomorphisms. In a pair $(\mathcal A,k) \in \struc[\tau] \times \mathbb
N$ the number~$k$ is, of course, the \emph{parameter value} of the pair.
Flum and Grohe now propose to describe such problems
slicewise using formulas. Since this will be the only way in which we
describe problems, we will drop the ``slicewise'' in the phrasings and
just say that a computable family $(\phi_k)_{k \in \mathbb N}$ of 
formulas \emph{describes} a problem $Q \subseteq
\struc[\tau] \times \mathbb N$  
if for all $(\mathcal A,k) \in \struc[\tau] \times \mathbb N$ we have
$(\mathcal A,k) \in Q$ if, and only if, $\mathcal A \models
\phi_k$. One can also define a purely logical notion of
reductions between two problems $Q$ and~$Q'$, but we will need this
notion only inside the proof of Theorem~\ref{thm-hitting} and
postpone the definition till then.
%

For a class $\Phi$ of computable families $(\phi_k)_{k \in \mathbb N}$,
let us write $\Class X \Phi$ for the class of all parameterized problems
that are described by the members of~$\Phi$ (we chose ``$\Class X$'' to
represent a ``slicewise'' description, which seems to be in good
keeping with the usual use of $\Class X$ in other classes such as $\Class{XP}$
or~$\Class{XL}$). For instance, the mentioned characterization of
$\Class{FPT}$ in logical terms by Flum and Grohe can be written as
$\Class{FPT} = \Class X\{(\phi_k)_{k\in\mathbb N} \mid \phi_k \in
\Class{FO}[\op{lfp}], \textstyle \max_k \operatorname{qr} (\phi_k) < \infty\}$.

We remark that instead of describing parameterized problems using
families, a more standard and at the same time more flexible way is to
use reductions to model checking problems. Clearly, if a family
$(\phi_k)_{k \in \mathbb N}$ of $\mathcal 
L$-formulas describes $Q \subseteq \struc[\tau] \times \mathbb
N$, then there is a very simple parameterized reduction from~$Q$ to the model
checking problem $\PLang[\phi]{mc}(\mathcal L)$, where the input is a
pair $(\mathcal A, \operatorname{num}(\phi))$ and the question is
whether both $\mathcal A \models \phi$ and $\phi \in \mathcal L$ hold.
(The function $\operatorname{num}$ encodes mathematical objects like
$\phi$ or later tuples like $(\phi,\delta)$ as unique natural numbers.) The
reduction simply maps a pair $(\mathcal A,k)$ to $(\mathcal 
A,\operatorname{num}(\phi_k))$. Even more interestingly, without going
into any of the technical details, it is also not hard to see that as long as a
reduction is sufficiently simple, the reverse implication holds, that
is, we can replace a reduction to the model checking problem by a
family of formulas that describe the problem. We can, thus, use
whatever formalism seems more appropriate for the task at hand and  --
as we hope that this paper shows -- it is sometimes quite natural to write
down a family that describes a problem.

\paragraph*{Parameterized Circuits.} For our descriptive setting,
we need to slightly adapt the definition of the circuit classes
$\Para\Class{AC}^0$ and $\Para\Class{AC}^{0\uparrow}$ from
\cite{BannachST2015,BannachT2016}: Let us say that a problem $Q
\subseteq \struc[\tau]\times\mathbb N$ is in
$\Para\Class{AC}^0$, if there is a family $(C_{n,k})_{n,k \in \mathbb
  N}$ of AC-circuits (Boolean circuits with unbounded fan-in) such
that for all $(\mathcal A,k) \in
\struc[\tau]\times\mathbb N$ we have, first, $(\mathcal
A,k) \in Q$ if, and only if, $C_{|x|,k}(x)=1$ where $x$ is
a binary encoding of $\mathcal A$; second, the size of $C_{n,k}$ is at
most $f(k)\cdot n^{c}$ for some computable function~$f$; third, the
depth of $C_{n,k}$ is bounded by a constant; and, fourth, the circuit
family satisfies a \textsc{dlogtime}-uniformity condition. The class
$\Para\Class{AC}^{0\uparrow}$ 
is defined the same way, but the depth may be $g(k)$ for some
computable~$g$ instead of only~$O(1)$. The following fact and theorem
show how these two 
circuit classes are closely related to descriptions of parameterized
problems using formulas: 

\begin{fact}[\cite{ChenFH2017}]\label{fact-ac0}
  $ \Para\Class{AC}^0 = \Class X\bigl\{(\phi_k)_{k\in\mathbb N} \bigm| \phi_k \in
  \Class{FO}[+,\times], \textstyle \max_k
  \operatorname{qr} (\phi_k) < \infty\bigr\}$.
\end{fact}

\begin{theorem}\label{thm-ac0up}
  $\Para\Class{AC}^{0\uparrow} = \Class X\bigl\{(\phi_k)_{k\in\mathbb N} \bigm| \phi_k \in
  \Class{FO}[+,\times], \textstyle \max_k \left|\operatorname{bound} (\phi_k)\right| < \infty\bigr\}$.
\end{theorem}
\begin{proof}
  The basic idea behind the proof is quite ``old'': we need to
  establish links between 
  circuit depth and size and the number of variables used in a 
  formula -- and such links are well-known, see for
  instance~\cite{Vollmer1999}: The \emph{quantifier
    rank} of a first-order formula naturally corresponds to the
  \emph{depth} of a circuit that solves the model checking problem for
  the formula. The \emph{number of variables} corresponds to the
  \emph{exponent of the polynomial} that bounds the size of the
  circuit (the paper~\cite{Immerman1991} is actually entitled
  $\Class{DSPACE}[n^k] = \Class{VAR}[k+1]$). One thing that is usually
  not of interest (because only one 
  formula is usually considered) is the fact that the \emph{length} of
  the formula is linked \emph{multiplicatively} to the size of the
  circuit.

  In detail, suppose we are given a problem $Q \subseteq
  \struc[\tau] \times \mathbb N$ with $Q \in
  \Para\Class{AC}^{0\uparrow}$ via a circuit family $(C_{n,k})_{n,k \in
    \mathbb N}$ of depth $g(k)$ and size $f(k) n^c$. For a fixed~$k$,
  we now need to construct a formula $\phi_k$ that correctly decides
  the $k$-th slice. In other words, we need a
  $\Class{FO}[+,\times]$-formula $\phi_k$ whose finite models are
  exactly those on which the family $(C_{n,k})_{n \in \mathbb N}$
  (note that $k$ no longer indexes the family) evaluates to~$1$ (when
  the models are encoded as bitstrings). It is well-known how such a
  formula can be constructed, see for instance \cite{Vollmer1999}, we
  just need a closer look at how the quantifier rank and number of
  variables relate to the circuit depth and size.

  The basic idea behind the formula $\phi_k$ is the following: The
  circuit has $f(k) n^c$ gates and we can ``address'' these gates
  using $c$ variables (which gives us $n^c$ possibilities) plus a
  number $i \in \{1,\dots, f(k)\}$ (which gives us $f(k) \cdot n^c$
  possibilities). Since for fixed~$k$ the number 
  $f(k)$ is also fixed, it is permissible that the formula $\phi_k$
  contains $f(k)$ copies of some subformula, where each subformula
  handles another value of~$i$. The basic idea is then to start with 
  formulas~$\psi^0_i$ for $i\in \{1,\dots,f(k)\}$, each of which has
  $c$ free variables, so that $\psi^0_i(x_1,\dots,x_c)$ is true
  exactly if the tuple $(x_1,\dots,x_c,i)$ represents an input gate
  set to~$1$. At this point, the uniformity condition basically tells
  us that such formulas can be constructed and that they all have a
  fixed quantifier rank. Next, we construct formulas
  $\psi^1_i(x_1,\dots,x_c)$ that are true if $(x_1,\dots,x_c,i)$
  addresses a gate for which the input values are all already computed
  by the $\psi^0_j$ and that evaluates to~$1$. Next, formulas
  $\psi^2_i$ are constructed, but, now, we can reuse the variables
  used in the $\psi^0_j$. In this way, we finally build formulas
  $\smash{\psi^{g(k)}_i}$ and apply it to the ``address'' of the output gate. All
  told, we get a formula whose quantifier rank is $c \cdot g(k) + O(1)$ and
  in which at most $2c + O(1)$ variables are used (note that the size
  of the formula depends on $f(k)$). Clearly, this means
  that the family $(\phi_k)_{k\in\mathbb N}$ created in this way does,
  indeed, only use a bounded number of variables (namely $O(c)$ many)
  and decides~$Q$.

  For the other direction, suppose  $(\phi_k)_{k\in\mathbb N}$
  describes $Q$ and that all $\phi_k$ contain at most~$v$ variables
  (since they contain no free variables, this is same as the number of
  bound variables). Clearly, we may assume that the $\phi_k$ are in
  negation normal form. We may also assume that they are \emph{flat,}
  by which we mean that they contain no subformulas of 
  the form $(\alpha \lor \beta) \land \gamma$ or $\alpha \land (\beta
  \lor \gamma)$: using the distributive laws of
  propositional logic, any first-order formula can be turned into an
  equivalent flat formula with the same number of variables and the
  same quantifier rank. Lastly, we may assume that the
  $\op{succ}$ function symbol is only used in atoms of the form $x =
  \op{succ}^s(0)$ for some variable~$x$ and some number~$s$: We can
  replace for instance $E\, \op{succ}^6(x) \op{succ}^3(y)$ by the
  equivalent formula $\exists x'\exists x'' \exists y' \exists
  y''(x'=\op{succ}^6(0) \land \op{add} x x' x'' \land y' =
  \op{succ}^3(0) \land \op{add} y y' y'' \land E x'' y'')$ without
  raising the number of variables and the quantifier rank by more than
  $4$ (or, in general, by more than the constant $2 \cdot
  \operatorname{arity}(\tau)$). 

  As before, it is now known that for each $\phi_k$
  there is a family $(C_{n,k})_{n \in \mathbb N}$ that evaluates
  to~$1$ exactly on the (encoded) models of~$\phi_k$. These circuits
  are constructed as follows: While $\phi_k$ has no free variables, a
  subformula $\psi$ of $\phi_k$ can have up to $v$ free variables. For
  each such subformula, the circuits use $n^v$ gates to keep track of
  all assignments to these~$v$ variables that make the subformula
  true. Clearly, this is relatively easy to achieve for literals in a
  constant number of layers, including literals of the form $x =
  \op{succ}^s(0)$ since $s$ is a fixed number depending only
  on~$k$. Next, if a formula is of the form  
  $\bigwedge_i \alpha_i$ and for some assignment we have one gate for
  each $\alpha_i$ that tells us whether it is true, we can feed all
  these wires into one $\land$-gate. We can take care of a formula of
  the form $\bigvee_i \alpha_i$ in the same way -- and note that in a
  flat formula there will be at most one alternation from
  $\bigwedge$ to $\bigvee$ before we encounter a quantifier. Now, for
  subformulas of the form $\exists x\, \phi$, the correct values for
  the $n^{v-1}$ gates can be obtained by a big $\lor$-gate attached to
  the outputs from the gates for $\phi$. Similarly, $\forall x\, \phi$
  can be handled using a big $\land$-gate.

  Based on these observations, it is now possible to build a circuit
  of size $|\phi_k| n^v$ and depth 
  $O(\operatorname{qr}(\phi_k))$. In particular, the resulting overall
  circuit family has a depth that depends only on the parameter (since
  the quantifier rank can be at most $|\phi_k|$, which depends only
  on~$k$) and has a size of at most $f(k)\cdot n^{c}$ for $f(k) =
  |\phi_k|$. It can also be shown that the necessary uniformity
  conditions are satisfied.

  We remark that the above proof also implies Fact~\ref{fact-ac0},
  namely for $g(k) = O(1)$ for the first direction and for
  $\operatorname{qr}(\phi_k) = O(1)$ for the second direction.
\end{proof}

\section{Syntactic Properties Allowing Color Coding}\label{section-cc}

The color coding technique~\cite{AlonYZ95} is a powerful method from
parameterized complexity theory for ``discovering small objects'' in
larger structures. Recall the example from the introduction: While
finding $k$ disjoint triangles in a graph is difficult in general, it
is easy when the graph is colored with $k$ colors and the objective is
to find for each color one triangle having this color. The idea behind
color coding is to reduce the (hard) uncolored version to the (easy)
colored version by \emph{randomly} coloring the graph and then
``hoping'' that the coloring assigns a different color to each
triangle. Since the triangles are ``small objects,'' the probability
that they do, indeed, get different colors depends only on~$k$. Even
more importantly, Alon et al. noticed that we can derandomize
the coloring procedure simply by coloring each vertex by its hash
value with respect to a simple family of universal hash functions that
only use addition and multiplication~\cite{AlonYZ95}. This idea is beautiful and
works surprisingly well in practice~\cite{HuffnerWZ08}, but using the method inside
proofs can be tricky: On the one hand, we need to ``keep the set sizes
under control'' (they must stay roughly logarithmic in size) and we
``need to actually identify the small set based just on its random
coloring.'' Especially for more complex proofs this can lead to rather
subtle arguments.

In the present section, we identify \emph{syntactic} properties of 
formulas that guarantee that the color coding technique can be
applied. The property is that the colors (the predicates~$C_i$ in
the formulas) are not in the scope of a universal quantifier (this
restriction is necessary, as the example of the formula describing
3-colorability shows). 

As mentioned already in the introduction, the main ``job'' of the
colors in proofs based on color coding is to ensure that vertices of a
graph are different from other vertices. This leads us to the idea of
focusing entirely on the notion of distinctness in the second half of
this section. This time, there will be syntactic properties of
existentially bounded first-order variables that will allow us to
apply color coding to them.

\subsection{Formulas With Color Predicates}

In graph theory, a \emph{coloring} of a graph can either refer to an
arbitrary assignment that maps each vertex to a color or to such an
assignment in which vertices connected by an edge must get different
colors (sometimes called \emph{proper} colorings). For our purposes,
colorings need not be proper and are thus partitions of the vertex set into \emph{color classes.} From
the logical point of view, each color class can be represented by a
unary predicate. A \emph{$k$-coloring of
  a $\tau$-structure $\mathcal A$} is a structure $\mathcal B$ over
the signature $\tau_{k\text{\normalfont-}\mathrm{colors}} = \tau \cup
\{C^1_1,\dots,C^1_k\}$, where the $C_i$ are fresh unary relation
symbols, such that $\mathcal A$ is the $\tau$-restriction of $\mathcal
B$ and such that the sets $C^{\mathcal   B}_1$ to $C^{\mathcal B}_k$
form a partition of the universe~$|\mathcal A|$ of~$\mathcal A$.  

Let us now formulate and prove the first
syntactic version of color coding. An example of a possible formula
$\phi$ in the theorem is $\bigwedge_{i=1}^k \exists x \exists y \exists z
(E x y \land Eyz \land Exz \land C_i x \land C_i y \land C_i z)$, 
for which the theorem tells us that there is a formula $\phi'$ of
constant quantifier rank that is true exactly when there are pairwise
disjoint sets~$C_i$ that make $\phi$ true. 

\begin{theorem}\label{thm-cc-intro}
  Let $\tau$ be an arithmetic signature and let $k$ be a number. For
  each first-order
  $\tau_{k\text{\normalfont-}\mathrm{colors}}$-sentence $\phi$ in
  negation normal form in which no~$C_i$ is inside a universal scope,
  there is a $\tau$-sentence $\phi'$ such that: 
  \enumerate
  \item
    For all $\mathcal A \in \struc[\tau]$ we have $\mathcal A \models
    \phi'$ if, and only if, there is a $k$-coloring $\mathcal
    B$ of $\mathcal A$ with $\mathcal B \models \phi$.
  \item
    $\operatorname{qr}(\phi') = \operatorname{qr}(\phi) + O(1)$.
  \item 
   $\left|\operatorname{bound}(\phi')\right| =
   \left|\operatorname{bound}(\phi)\right| + O(1)$. 
   \endenumerate
\end{theorem}
(Let us clarify that $O(1)$ represents a global constant that is independent
of $\tau$ and~$k$.)

\begin{proof}
  Let $\tau$, $k$, and $\phi$ be given as stated in the
  theorem. If necessary, we modify $\phi$ to ensure that there is no
  literal of the form $\neg C_i x_j$, by replacing each such literal
  by the equivalent $\bigvee_{l\neq i} C_l x_j$. After this
  transformation, the $C_i$ in $\phi$ are neither in the scope of
  universal quantifiers nor of negations -- and this is also true for
  all subformulas~$\alpha$ of~$\phi$. We will now show by structural
  induction that all these subformulas (and, hence, also $\phi$) have
  two \emph{semantic} properties, which we call the \emph{monotonicity
    property} and the \emph{small witness property} (with respect to
  the $C_i$). Afterwards, we will show that these two properties allow
  us to apply the color coding technique.

  \subparagraph*{Establishing the Monotonicity and Small Witness Properties.}

  Some notations will be useful: Given a $\tau$-structure $\mathcal A$
  with universe $A$ and given sets $A_i \subseteq A$
  for $i \in \{1,\dots,k\}$, let us write $\mathcal 
  A \models \phi (A_1,\dots, A_k)$ to indicate that
  $\mathcal B$ is a model of $\phi$ where $\mathcal B$ is the
  $\tau_{k\text{\normalfont-}\mathrm{colors}}$-structure with universe
  $A$ in which all symbols from $\tau$ are interpreted as in~$\mathcal
  A$ and in which the symbol~$C_i$ is interpreted as~$A_i$, that is,
  $C_i^{\mathcal B} = A_i$. Subformulas $\gamma$ of~$\phi$ may
  have free variables and suppose that $x_1$ to $x_m$ are the free variables
  in~$\gamma$ and let $a_i \in A$ for $i\in \{1,\dots,m\}$. We write
  $\mathcal A \models \gamma (A_1,\dots, 
  A_k, a_1,\dots,a_m)$ to indicate that $\gamma$ holds in the
  just-described structure~$\mathcal B$ when each $x_i$ is interpreted
  as~$a_i$.   

  \begin{definition}
    Let $\gamma$ be a $\tau_{k\text{\normalfont-}\mathrm{colors}}$-formula with free variables $x_1$ to
    $x_m$. We say that $\gamma$ has the \emph{monotonicity and the
      small witness properties with respect to the $C_i$} if for all
    $\tau$-structures $\mathcal A$ with universe~$A$ and all values
    $a_1,\dots,a_m \in A$ the following holds:
    \begin{enumerate}
    \item \emph{Monotonicity property:} Let $A_1,\dots,A_k \subseteq A$
      and $B_1,\dots,B_k \subseteq A$ be sets with $A_i \subseteq B_i$
      for all $i \in \{1,\dots,k\}$. Then  
      $\mathcal A \models \gamma(A_1,\dots,A_k,\penalty 0 a_1,\dots,a_m)$ implies
      $\mathcal A \models \gamma(B_1,\dots,B_k,\penalty0 a_1,\dots,a_m)$.  
    \item \emph{Small witness property:} If there are any pairwise
      disjoint sets $B_1,\dots,B_k \subseteq A$  such that $\mathcal A \models
      \gamma(B_1,\dots,B_k,a_1,\dots,a_m)$, then 
      there are sets $A_i \subseteq B_i$ whose
      sizes $|A_i|$  depend only on $\gamma$ for $i \in \{1,\dots,k\}$,
      such that $\mathcal A \models \gamma(A_1,\dots,A_k,a_1,\dots,a_m)$.
    \end{enumerate}
  \end{definition}

  We now show that $\phi$ has these two properties (for $m=0$). For
  monotonicity, just note that the~$C_i$ are not in the scope of any
  negation and, thus, if some $A_i$ make $\phi$ true,
  so will all supersets~$B_i$ of the~$A_i$.

  To see that the small witness property holds, we argue by
  structural induction: If $\phi$ is any formula that does not involve
  any $C_i$, then $\phi$ is true or false independently of the $B_i$
  and, in particular, if it is true at all, it is also true for $A_i =
  \emptyset$ for $i \in \{1,\dots,k\}$. If $\phi$ is the atomic
  formula $C_i x_j$, then setting $A_i = \{a_j\}$ and $A_{i'} =
  \emptyset$ for $i' \neq i$ makes the formula true.

  If $\phi = \alpha \land \beta$, then $\alpha$ and $\beta$ have the
  small witness property by the induction hypothesis. Let
  $B_1,\dots,B_k\subseteq A$ make $\phi$ hold in $\mathcal A$. Then
  they also make both $\alpha$ and $\beta$ hold in 
  $\mathcal A$. Let $A^\alpha_1,\dots,A^\alpha_k\subseteq A$ with
  $A^\alpha_i \subseteq B_i$ be the witnesses for $\alpha$ and let
  $A^\beta_1,\dots,A^\beta_k\subseteq A$ be the witnesses 
  for~$\beta$. Then by the monotonicity property, $A^\alpha_1 \cup
  A^\beta_1,\dots,A^\alpha_k \cup A^\beta_k$ makes both $\alpha$ and
  $\beta$ true, that is
  \begin{align*}
    \mathcal A \models \alpha(A^\alpha_1
    \cup A^\beta_1,\dots,A^\alpha_k \cup A^\beta_k, a_1,\dots,a_m)
  \end{align*}
  and the same holds for $\beta$. Note that $A^\alpha_i \cup
  A^\beta_i \subseteq B_i$ still holds and that they have sizes
  depending only on $\alpha$ and~$\beta$ and thereby on~$\phi$.

  For $\phi = \alpha \lor \beta$ we can argue in exactly the same way
  as for the logical and.

  The last case for the structural induction is $\phi = \exists x_m
  (\alpha)$. Consider pairwise disjoint $B_1,\dots,B_k \subseteq A$
  that make $\phi$ true. Then there is a value $a_m \in A$ such that
  $\mathcal A \models
  \alpha(B_1,\dots,B_k,a_1,\dots,a_m)$. Now, since
  $\alpha$ has the small witness property by the induction hypothesis,
  we get $A_i \subseteq B_i$ of size depending on $\alpha$ for which
  we also have $\mathcal A \models
  \alpha(A_1,\dots,A_k,a_1,\dots,a_m)$. But then,
  by the definition of existential quantifiers, these $A_i$ also witness $\mathcal A \models
  \exists x_m \phi(A_1,\dots,A_k,\penalty0 a_1,\dots,a_{m-1})$. (Observe
  that this is the point where the argument would \emph{not} work for
  a universal quantifier: Here, for each possible value of $a_m$ we
  might have a different set of~$A_i$'s as witnesses and their union
  would then no longer have small size.)

  \subparagraph*{Applying Color Coding.}

  Our next step in the proof is to use color coding to produce the
  partition. First, let us recall the basic lemma on universal hash
  functions formulated below in a way equivalent to \cite[page
  347]{FlumG2006}:  

  \begin{lemma}\label{lemma-color-coding}
    There is an $n_0 \in \mathbb N$ such that for all $n \ge n_0$ and
    all subsets $X \subseteq \{0,\dots,n-1\}$ there  exist a prime $p
    < |X|^2 \log_2 n$ and a number $q < p$ such that the function $h_{p,q}(m)
    = (q\cdot m \mathbin{\mathrm{mod}} p) \mathbin{\mathrm{mod}} |X|^2$ is
    injective on~$X$. 
  \end{lemma}

  As has already been observed by Chen et al.~\cite{ChenFH2017}, if we
  set $k = |X|$ we can easily express the computation underlying
  $h_{p,q}\colon \{0,\dots,n-1\} \to \{0,\dots,k^2-1\}$ using a fixed
  $\Class{FO}[+,\times]$-formula $\rho(k,p,q,x,y)$. That is, if we
  encode the numbers $k,p,q,x,y \in \{0,\dots, n-1\}$ as corresponding
  elements of the universe with respect to the ordering of the universe, then
  $\rho(k,p,q,x,y)$ holds if, and only if, $h_{p,q}(x) = y$. Note that
  the $p$ and $q$ from the lemma could exceed $n$ for very large~$X$
  (they can reach up to $n^2 \log_2 n \le n^3$), but, first, this
  situation will not arise in the following and, second, this could be
  fixed by using three variables to encode $p$ and three variables to
  encode $q$. Trivially, $\rho(k,p,q,x,y)$ has some constant
  quantifier rank (the formula explicitly constructed by Chen et al.\
  has $\operatorname{qr}(\rho) = 9$, assuming $k^2 < n$).

  Next, we will need the basic idea or ``trick'' of Alon et
  al.'s~\cite{AlonYZ95} color coding technique: While for appropriate
  $p$ and $q$ the function $h_{p,q}$ will ``just'' be injective on
  $\{0,\dots,k^2-1\}$, we actually want a function that maps each
  element $x\in X$ to a specific element (``the color of $x$'') of
  $\{1,\dots,k\}$. Fortunately, this is easy to achieve by
  concatenating $h_{p,q}$ with an appropriate function $g \colon
  \{0,\dots,k^2-1\} \to \{1,\dots,k\}$. 
  
  In detail, to construct $\phi'$ from the claim of the theorem, we 
  construct a family of formulas $\phi^g(p,q)$ where $p$
  and $q$ are new free variables and the formulas are indexed by all
  possible functions $g \colon \{0,\dots,k^2-1\} \to \{1,\dots,k\}$:
  In $\phi$, replace every occurrence of $C_i x_j$ by the following
  formula $\pi_i^g(p,q,x_j)$:
  \begin{align*}
    \textstyle\bigvee_{y\in\{0,\dots,k^2-1\},g(y) = i}
    \exists \hat k \exists \hat y \bigl(\op{succ}^k(0) = \hat k \land
    \op{succ}^y(0) = \hat y \land \rho(\hat k,p,q,x_j,\hat y)\bigr)
  \end{align*}
  where $\hat k$
  and $\hat y$ are fresh variables that we bind to the numbers $k$
  and~$y$ (if the universe is large enough). Note that the  
  formula $C_i x_j$ has $x_j$ as a free variable, while
  $\pi_i^g(p,q,x_j)$ additionally has $p$ and $q$ as 
  free variables. As an example, for the formula $\phi = \exists x (C_2 x
  \lor \exists y C_5 y)$ we would have $\phi^g = \exists x(\pi_2^g
  (p,q,x) \lor \exists y \pi_5^g(p,q,y))$. Clearly, each $\phi^g$ has the
  property $\operatorname{qr}(\phi^g) = \operatorname{qr}(\phi) +
  O(1)$.

  The desired formula $\phi'$ is (almost) simply
  $\bigvee_{g:\{0,\dots,k^2-1\} \to \{1,\dots,k\}} \exists p
  \exists q (\phi^g(p,q))$. The ``almost'' is due to the fact that this
  formula works only for structures with a sufficiently large
  universe -- but by Lemma~\ref{lemma-eventually} it suffices to
  consider only this case. Let us prove that for every
  $\sigma$-structure $\mathcal A$ with universe $A = \{0,\dots,n-1\}$
  and $n \ge c$ for some to-be-specified constant~$c$, the following
  two statements are equivalent:
  \begin{enumerate}
  \item
    There is a $k$-coloring $\mathcal B$ of $\mathcal A$ with
    $\mathcal B \models \phi$.
  \item 
    $\mathcal A \models \bigvee_{g:\{0,\dots,k^2-1\}
      \to \{1,\dots,k\}} \exists p \exists q (\phi^g(p,q))$.
  \end{enumerate}
  Let us start with the implication of item~2 to~1. Suppose there is a
  function $g\colon\{0,\dots,\penalty0 k^2-1\} \to \{1,\dots,k\}$ and
  elements $p,q \in 
  \{0,\dots,n-1\}$ such that  $\mathcal A \models \phi^g(p,q)$. We
  define a partition $A_1  \mathbin{\dot\cup} \cdots \mathbin{\dot\cup} A_k = 
  A$ by $A_i = \{ x \in A \mid g(h_{p,q}(x)) = i\}$. In
  other words, $A_i$ contains all elements of $A$ that are first
  hashed to an element of $\{0,\dots,k^2-1\}$ that is then mapped
  to~$i$ by the function~$g$. Trivially, the $A_i$ form a partition of
  the universe~$A$.

  Assuming that the universe size is sufficiently large, namely for
  $k^2 \log_2 n < n$, inside $\phi^g$ all uses of $\rho(\hat
  k,p,q,x,\hat y)$ will have the property that $\mathcal A \models
  \rho(\hat k,p,q,x,\hat y)$ if, and only if, $h_{p,q}(x)
  = \hat y$. Clearly, there is a constant $c$ depending
  only on $k$ such that for all $n>c$ we have $k^2 \log_2 n < n$.

  With the property established, we now see that $\pi_i^g(p,g,x_j)$
  holds inside the formula $\phi^g$ if, and only if, the
  interpretation of $x_j$ is an element of $A_i$. This means that if
  we interpret each $C_i$ by $A_i$, then we get $\mathcal A \models
  \phi(A_1,\dots,A_k)$ and the $A_i$ form a partition of the
  universe. In other words, we get item~1.

  Now assume that item~1 holds, that is, there is a partition
  $B_1 \mathbin{\dot\cup} \cdots \mathbin{\dot\cup} B_k = A$ with
  $\mathcal A \models \phi(B_1,\dots,B_k)$. We must show that there
  are a $g \colon \{0,\dots,k^2-1\} \to \{1,\dots,k\}$ and $p,q \in A$
  such that $\mathcal A \models \phi^g(p,q)$.

  At this point, we use the small witness property that we established
  earlier for the partition. By this property there are pairwise
  disjoint sets $A_i \subseteq A$ such that, first, $|A_i|$ depends only
  on $\phi$ and, second, $\mathcal A \models \phi(A_1,\dots,A_k)$. Let $X =
  \bigcup_{i=1}^k A_i$. Then $|X|$ depends only on $\phi$ and let $s_\phi$
  be a $\phi$-dependent upper bound on this size. By the universal
  hashing lemma, there are now $p$ and $q$ such that $h_{p,q} \colon
  \{0,\dots,n-1\} \to \{0,\dots,s_\phi^2 
  -1\}$ is injective on~$X$. But, then, we can set $g \colon
  \{0,\dots,s_\phi^2 -1\} \to \{1,\dots,k\}$ to $g(v) = i$ if there is
  an $x \in A_i$ with $h_{p,q}(x) = v$ and setting $g(v)$ arbitrarily
  otherwise. Note that this is, indeed, a valid definition of~$g$
  since $h_{p,q}$ is injective on~$X$.

  With these definition, we now define the following sets $D_1$ to
  $D_k$: Let $D_i = \{x \in A\mid g(h_{p,q}(\hat x)) =
  i\}$ where $\hat x$ is the index of $x$ in $A$ with respect to the
  ordering (that is, $\hat x = |\{ y \in A \mid y <^{\mathcal A} x\}|$
  and for the special case that $A = \{0,\dots,n-1\}$ and that $<^{\mathcal 
    A}$ is the natural ordering, $\hat x = x$). Observe
  that $D_i \supseteq A_i$ holds for all $D_i$ and that 
  the $D_i$ form a partition of the universe~$A$. By the monotonicity
  property,  $\mathcal A \models \phi(A_1,\dots,A_k)$ implies
  $\mathcal A \models \phi(D_1,\dots,D_k)$. However, by definition of
  the $D_i$ and of the formulas $\pi_i^g$, for a sufficiently large
  universe size~$n$ (namely $s_\phi^2 \log_2 n < n$), we now also
  have $\mathcal A \models \phi^g(p,q)$, which in turn implies
  $\mathcal A \models \bigvee_g \exists p \exists q \phi^g$.
\end{proof}

In the theorem we assumed that $\phi$ is a sentence to keep the
notation simple, both the theorem and later theorems still hold when
$\phi(x_1,\dots,x_n)$ has free variables $x_1$ to~$x_n$. Then there is
a corresponding $\phi'(x_1,\dots,x_n)$ such that first item becomes
that for all $\mathcal A \in \struc[\tau]$ and all $a_1,\dots,a_n \in
|\mathcal A|$ we have 
$\mathcal A \models \phi'(a_1,\dots,a_n)$ if, and only if, there is a
$k$-coloring $\mathcal B$ of~$\mathcal A$ with $\mathcal B \models
\phi(a_1,\dots,a_n)$. Note that the syntactic transformations in the
theorem do not add dependencies of universal quantifiers on the free
variables. 

\subsection{Formulas With Weak Quantifiers}

If one has a closer look at proofs based on color coding, one cannot
help but notice that the colors are almost exclusively used to ensure that
certain vertices in a structure are distinct from certain other
vertices: recall the introductory example $\bigwedge_{j=1}^k
\exists x \exists y \exists z (E x y \land Eyz \land Exz \land C_{3j-2}
x \land C_{3j-1} y \land C_{3j} z)$, which describes the triangle
packing problem when we require that the $C_i$ form a partition of the 
universe. Since the $C_i$ are only used to ensure that the many
different $x$, $y$, and $z$ are different, we already rewrote the
formula in \eqref{eq:tri} as $\exists x_1 \cdots \exists x_{3k}\textstyle
\bigwedge_{i\neq j} x_i \neq x_j \land \bigwedge_{j=1}^k \exists x
\exists y \exists z (E xy \land Eyz \land Exz \land x_{3j-2} = x
\land x_{3j-1} =y \land x_{3j} =z)$. While this rewriting gets rid of
the colors and moves us back into the familiar territory of simple 
first-order formulas, the quantifier rank and the number of variables
in the formula have now ``exploded'' (from the constant $3$ to the
parameter-dependent value $3k+3$) -- which is exactly what we need to
avoid in order to apply Fact~\ref{fact-ac0} or
Theorem~\ref{thm-ac0up}.

We now define a syntactic property that the $x_i$ have that allows us
to remove them from the formula and, thereby, to arrive at a family of
formulas of constant quantifier rank. For a (sub)formula $\alpha$ of
the form $\forall d(\phi)$ or 
$\exists d(\phi)$, we say that $d$ \emph{depends} on all free
variables in~$\phi$ (at the position of $\alpha$ in a larger
formula). For instance, in $Exy \land \forall b (Exb \land \exists z (E yz)) \land
\exists b (E xx)$, the variable $b$ depends on $x$ and $y$ at its
first binding ($\forall b$) and on~$x$ at the second binding~($\exists
b$).
\begin{definition}\label{def-weak}
  We call the leading quantifier in a formula $\exists x (\phi)$ in
  negation normal form \emph{strong} if 
  \begin{enumerate}
  \item 
    some universal binding inside~$\phi$ depends on~$x$ or
  \item
    there is a subformula $\alpha \land \beta$ of~$\phi$ such that
    both $\alpha$ and~$\beta$ contain $x$ in literals that are not of
    the form $x \neq y$ for some variable~$y$. 
  \end{enumerate}
  If neither of the above hold, we call the quantifier \emph{weak}. 
  The \emph{strong quantifier rank} $\operatorname{strong-qr}(\phi)$
  is the quantifier rank of~$\phi$, where weak quantifiers are
  ignored; $\operatorname{strong-bound}(\phi)$ contains all
  variables of $\phi$ that are bound by non-weak quantifiers.
\end{definition}
(Later on we extend the definition to the dual notion of weak
\emph{universal} quantifiers, but for the moment let us only call
existential quantifiers weak.)

We place a dot on the variables bound by weak quantifiers to make
them easier to spot. For example, in $\phi = \exists x \exists y \exists \dot
z (R xx\dot z\dot z \land x \neq y \land y \neq \dot z \land Px \land
\forall w\, E w yy)$ the quantifier $\exists \dot z$ is weak, but
neither are $\exists x$ (since $x$ is used in two literals joined by a
conjunction, namely in $R xx\dot z\dot z$ and $P x$) nor $\exists y$
(since $w$ depends on~$y$ in $\forall w \, Ewyy$). We have
$\operatorname{qr}(\phi) = 4$, but $\operatorname{strong-qr}(\phi)
=3$, and $\operatorname{bound}(\phi) = \{x,y,\dot z\}$, but
$\operatorname{strong-bound}(\phi) = \{x,y\}$. 

Admittedly, the definition of weakness is a bit technical, but note that
there is a rather simple sufficient condition for a variable~$x$ to be
weak: If it not used in universal binding and used in only one literal
that is not an inequality, then $x$ is weak. This condition almost
always suffices for identifying the weak variables, although there are
of course exceptions like $\exists \dot x(P \dot x \lor Q \dot x)$.

\begin{theorem}\label{theorem-strong-character}
  Let $\tau$ be an arithmetic signature. Then for every $\tau$-formula
  $\phi$ in negation normal form there is a $\tau$-formula $\phi'$ such that
  \enumerate
  \item
    $\phi'$ is equivalent to $\phi$ on finite structures,
  \item
    $\operatorname{qr}(\phi') = 3\cdot\operatorname{strong-qr}(\phi)+O(\operatorname{arity}(\tau))$,
    and
  \item
    $\left|\operatorname{bound}(\phi')\right| =
    \left|\operatorname{strong-bound}(\phi)\right| + O(\operatorname{arity}(\tau))$. 
  \endenumerate
\end{theorem}

Before giving the detailed proof, we briefly sketch the overall idea:
Using simple syntactic transformations, we can ensure that all weak
quantifiers follow in  
blocks after universal quantifiers. We can also ensure that inequality
literals directly follow the blocks of weak quantifiers and are joined
by conjunctions. If the inequality literals following a block happen
to require that all weak variables from the block are different (that
is, if for all pairs $\dot x_i$ and $\dot x_j$ of different weak
variables there is an inequality $\dot x_i \neq \dot x_j$), then we
can remove the weak quantifiers $\exists \dot x_i$ and at the (single)
place where $\dot x_i$ is used, we use a color $C_i$ instead. For
instance, if $\dot x_i$ is used in the literal $\dot x_i = y$, we
replace the literal by $C_i y$. If $\dot x_i$ is used for instance in
$\neg E \dot x_i y$, we replace this by $\exists x(C_i x \land \neg E
xy)$. In this way, for each block we get an equivalent formula to
which we can apply Theorem~\ref{thm-cc-intro}. A more complicated
situation arises when the inequality literals in a block ``do not
require complete distinctness,'' but this case can also be handled by
considering all possible ways in which the inequalities can be
satisfied in parallel. In result, all weak
quantifiers get removed and for each block a constant number of new
quantifiers are introduced. Since each block follows a different
universal quantifier, the new total quantifier rank is at most the
strong quantifier rank times a constant factor; and the new number of
variables is only a constant above the number of original strong
variables.   

\begin{proof}
  Let $\phi$ be given.
  We first apply a number of simple syntactic
  transformations to move the weak quantifiers directly behind
  universal quantifiers and to move inequality literals directly
  behind blocks of weak quantifiers. Then we show how sets of
  inequalities can be ``completed'' if necessary. Finally, we
  inductively transform the formula in such a way that 
  Theorem~\ref{thm-cc-intro} can be applied repeatedly.

  As a running example, we use the (semantically not very sensible,
  but syntactically interesting) formula
  \begin{align}
    \phi=\exists a \bigl(\exists \dot x (E a \dot x
    \land \exists \dot y(\dot y \neq \dot x)) \land
    \forall c \exists \dot x \exists \dot y (E \dot x \dot y \lor
    \exists z(\dot
    x \neq \dot y\land Pz\land Qc))\bigr) \label{eq-init}
  \end{align}
  and for each transformation we show how it applies to this example.

  \subparagraph*{Preliminaries.}

  It will be useful to require that all weak variables are
  different. Thus, as long as necessary, when a variable is
  bound by a weak quantifier and once more by another quantifier, replace the
  variable used by the weak quantifier by a fresh variable. Note that this
  may increase the number   of distinct (weak) variables in the
  formula, but we will get rid of   all of them later on anyway. From
  now on, we may assume that the weak variables are all distinct from
  one another and also from all other variables.

  It will also be useful to assume that $\phi$ starts with a universal
  quantifier. If this is not the case, replace 
  $\phi$ by the equivalent formula $\forall v (\phi)$ where $v$ is a
  fresh variable. This increases the quantifier rank by at
  most~$1$.

  Finally, it will also be useful to assume that the formula has been
  ``flatten'' as in the proof Theorem~\ref{thm-ac0up}: We use the
  distributive laws of propositional logic to 
  repeatedly replace subformulas of the form $(\alpha \lor \beta)
  \land \gamma$ by $(\alpha \land \gamma) \lor (\beta \land \gamma)$
  and $\alpha \land (\beta \lor \gamma)$ by $(\alpha \land \beta) \lor
  (\alpha \land \gamma)$. Note that this transformation does not
  change which variables are weak.
  
  For our running example, applying the described preprocessing yields:
  \begin{align*}
    \phi\equiv\forall v \exists a \bigl(\exists \dot x_1(E a \dot x_1 \land \exists \dot
    x_2 (\dot x_2 \neq \dot x_1)) \land
    \forall c \exists \dot x_3 \exists \dot x_4 (E \dot x_3 \dot x_4
    \lor \exists z(\dot
    x_3 \neq \dot x_4 \land Pz\land Qc))\bigr) 
  \end{align*}
  
  \subparagraph*{Syntactic Transformations I: Blocks of Weak Quantifiers.}

  The first interesting transformation is the following: We wish to
  move weak quantifiers ``as far up the syntax tree as possible.'' To
  achieve this, we apply the following equivalences as long as
  possible by always replacing the left-hand side (and also
  commutatively equivalent formulas) by the right-hand side:
  \begin{align*}
    \exists \dot x(\alpha) \land \beta & \equiv \exists \dot x(\alpha \land \beta), \\
    \exists \dot x(\alpha) \lor \beta & \equiv \exists \dot x(\alpha \lor \beta), \\
    \exists y \exists \dot x(\alpha) & \equiv \exists \dot x \exists y(\alpha). 
  \end{align*}
  Note that $\beta$ does not contain $\dot x$ since we  made all weak
  variables distinct and, of course, by $\exists y$ we mean a strong
  quantifier.

  Once the transformations have been applied exhaustively, all weak
  quantifiers will be directly preceded in $\phi$ by either a
  universal quantifier or another weak quantifier. This means that all
  weak quantifiers are now arranged in blocks inside~$\phi$, each
  block being preceded by a universal quantifier.
  \begin{align*}
    \phi\equiv\forall v \exists \dot x_1 \exists \dot x_2 \exists a
    \bigl(E a \dot x_1 \land \dot x_2 \neq \dot x_1 \land
    \forall c \exists \dot x_3 \exists \dot x_4 (E \dot x_3 \dot x_4
    \lor \exists z(\dot
    x_3 \neq \dot x_4 \land Pz\land Qc))\bigr) 
  \end{align*}

  \subparagraph*{Syntactic Transformations II: Weak and Strong Literals.}

  In order to apply color coding later on, it will be useful to have
  only three kinds of literals in $\phi$:
  \begin{enumerate}
  \item \emph{Strong literals} are literals that do not contain any
    weak variables.
  \item \emph{\!Weak equalities} are literals of the form $\dot x = y$
    involving exactly one strong variable that is existentially bound
    inside the weak variable's scope: $\exists \dot x (\dots \exists
    y(\dots \dot x = y \dots)\dots)$. 
  \item \emph{\!Weak inequalities} are literals of the form $\dot x \neq
    \dot y$ for two weak variables.
  \end{enumerate}
  Let us call all other kinds of literals \emph{bad.} This includes
  literals like $E \dot x \dot x$ or $E z \dot y$ that contain a
  relation symbol and some weak variables, but also inequalities $\dot
  x \neq y$ involving a weak and a strong variable, an equalities $\dot
  x = \dot y$ involving two weak variables, or an equality literal
  like the one in $\forall y \exists \dot x (\dot x= y)$. Finally,
  literals involving the successor function and weak variables are also
  bad. 

  In order to get rid of the bad literals, we will replace them by
  equivalent formulas that do not contain any bad literals. The idea
  is that we bind the variable or term that be wish to get rid of
  using a new existential quantifier. In order to avoid introducing
  too many new variables, for all of the following transformations we
  use the set of fresh variables $v_1$, $v_2$, and so on, where we may
  need more than one of these variables per literal, but will need no
  more than $O(\operatorname{arity}(\tau))$ (recall that
  $\operatorname{arity}(\tau)$ is the maximum arity of relation symbols in~$\tau$).
  
  Let us first get rid of the successor functions. If a bad literal
  $\lambda$ contains $\op{succ}^k(\mathring x)$ (where $\mathring x$
  indicates that $x$ may be strong or weak), we replace $\lambda$ by
  $\exists v_i \exists v_{i+1}(v_i = \op{succ}^k(0) \land
  \op{add}\mathring x v_i v_{i+1} \land
  \lambda[\op{succ}^k(\mathring x) \hookrightarrow v_{i+1}])$. Here,
  $\lambda [t_1 \hookrightarrow t_2]$ is our notation for the
  substitution of $t_1$ by~$t_2$ in~$\lambda$. The number~$i$ is chosen minimally so that $\lambda$
  contains neither $v_i$ nor $v_{i+1}$. Clearly, if we repeatedly
  apply this transformation to all literals containing the successor
  function, we get an equivalent formula in which no bad literal
  contains the successor function. Note that we use at most
  $2\operatorname{arity}(\tau)$ of the variables~$v_i$.

  Next, we get rid of the remaining bad literals, which are literals
  $\lambda$ that contain a weak variable $\dot x$, but are neither
  weak equalities not weak inequalities. This time, we replace
  $\lambda$ by $\exists v_i(v_i = \dot x \land \lambda[\dot x \hookrightarrow
  v_i])$ where, once more, $i$~is chosen minimally to avoid a name
  clash. Since this transformation reduces the number of weak
  variables in $\lambda$ and does not introduce a bad literal, sooner
  or later we will have gotten rid of all bad literals. Once more, for
  each literal we use at most $\operatorname{arity}(\tau)$ new
  variables from the~$v_i$. 

  Overall, we get that $\phi$ is equivalent to a formula without any
  bad literals in which we use at most $3\operatorname{arity}(\tau)$
  additional variables and whose quantifier rank is larger than that
  of~$\phi$ by at most $3\operatorname{arity}(\tau)$. Note
  that the transformation ensures that weak variables stay
  weak. Applied to our example formula, we get: 
  \begin{align*}
    \forall v \exists \dot x_1\exists \dot x_2 \exists a
    \bigl(
    &\exists v_1 (v_1 = \dot x_1 \land E a v_1) \land \dot x_2 \neq
      \dot x_1 \land {}\\
    &\forall c \exists \dot x_3 \exists \dot x_4 (\exists v_1(v_1 =
    \dot x_3 \land \exists v_2(v_2 = \dot x_4 \land E v_1 v_2)) \lor
      \exists z(\dot x_3 \neq \dot x_4 \land Pz\land Qc))\bigr) 
  \end{align*}
  
  \subparagraph*{Syntactic Transformations III: Accumulating Weak Inequalities.}

  We now wish to move all weak inequalities to the ``vicinity'' of the
  corresponding block of weak quantifiers. More precisely, just as we
  did earlier, we apply the following equivalences (interpreted once
  more as rules that are applied from left to right):
  \begin{align}
    (\dot x \neq \dot y \lor \alpha) \land \beta
    & \equiv
      (\dot x \neq \dot y \land \beta) \lor (\alpha \land \beta), \label{eq-or1}
    \\
    \exists x(\alpha \lor \beta)
    & \equiv
      \exists x(\alpha) \lor \exists x(\beta), \label{eq-or2}
    \\                                                   
    \exists z(\dot x \neq \dot y \land \alpha) & \equiv \dot x \neq \dot y \land \exists z(\alpha). \label{eq-inequ}
  \end{align}
  Note that these rules do not change which variables are weak.  
  When these rules can no longer be applied, the weak inequality
  are ``next'' to their quantifier block, that is, each subformula
  starting with weak quantifiers has the form 
  \begin{align*}
    \exists \dot x_{i_1} \cdots
    \exists \dot x_{i_k} \textstyle \bigvee_i \bigl(\bigl(\bigwedge_j
    \lambda_i^j\bigr) \land \alpha_i \bigr)
  \end{align*}
  where the $\alpha_i$ contain no weak
  inequalities while all $\lambda_i^j$ are weak inequalities.

  For our example formula, we get:
  \begin{align*}
    \phi\equiv\forall v \exists \dot x_1\exists \dot x_2 \bigl(  
    &\dot x_2 \neq
      \dot x_1 \land 
    \exists a \exists v_1 (v_1 = \dot x_1 \land E a v_1) \land {}\\
    &\forall c \bigl(\exists \dot x_3 \exists \dot x_4 (\exists v_1(v_1 =
    \dot x_3 \land \exists v_2(v_2 = \dot x_4 \land E v_1 v_2))) \lor{}\\
    &\phantom{\forall c \bigl(\exists \dot x_3 \exists \dot x_4(}(\dot x_3 \neq \dot x_4 \land \exists z(Pz\land Qc))\bigr)\bigr).
  \end{align*}

  Finally, we now swap each block of weak quantifiers with the
  following disjunction, that is, we apply the following equivalence
  from left to right:
  \begin{align*}
    \exists \dot x_{i_1} \cdots \exists \dot x_{i_k} \textstyle \bigvee_i \psi_i
    & \equiv \textstyle \bigvee_i \exists \dot x_{i_1} \cdots \exists \dot x_{i_k}\, \psi_i.
  \end{align*}
  If necessary, we rename weak variables to ensure once more that they
  are unique. 
  For our example, the different transformations yield:
  \begin{align*}
    \phi\equiv\forall v \exists \dot x_1\exists \dot x_2 \bigl(  
    &\dot x_2 \neq
      \dot x_1 \land 
    \exists a \exists v_1 (v_1 = \dot x_1 \land E a v_1) \land {}\\
    &\forall c \bigl(\exists \dot x_3 \exists \dot x_4 (\exists v_1(v_1 =
    \dot x_3 \land \exists v_2(v_2 = \dot x_4 \land E v_1 v_2))) \lor{}\\
    &\phantom{\forall c\bigl(} \exists \dot x_5 \exists \dot x_6(\dot x_5 \neq \dot x_6 \land \exists z(Pz\land Qc))\bigr)\bigr).
  \end{align*}
  Let us spell out the different $\psi_i$, $\lambda_i^j$, and
  $\alpha_i$ contained in the above formula: First, there is one block
  of weak variables ($\exists \dot x_1 \exists \dot x_2$) following
  $\forall v$ at the beginning. There is only a single $\psi_1$ for
  this block, which equals $(\bigwedge_{j=1}^1 \lambda_1^j) \land
  \alpha_1$ for $\lambda_1^1 = (\dot x_2 \neq \dot x_1)$ and $\alpha_1
  = \exists a \exists v_1 (v_1 = \dot x_1 \land E a v_1) \land \forall
  c (\dots)$. Second, there are two blocks of weak variables ($\exists
  \dot x_3 \exists \dot x_4$ and $\exists \dot x_5 \exists \dot x_6$)
  following~$\forall c$, which are followed by (new) formulas $\psi_1$
  and~$\psi_2$. The first is of the form $\psi_1 = (\bigwedge_{j=1}^0
  \lambda_1^j) \land \alpha_1$ and the second of the form $\psi_2 = (\bigwedge_{j=1}^1
  \lambda_2^j) \land \alpha_2$. There are no  $\lambda_1^j$ and we
  have $\alpha_1 = \exists v_1(v_1 = \dot x_3 \land \exists v_2(v_2 =
  \dot x_4 \land E v_1 v_2))$. We have $\lambda_2^1 = (\dot x_5 \neq
  \dot x_6)$ and we have $\alpha_2 = \exists z(Pz\land Qc)$.
  
  We make the following observation at this point: Inside
  each $\psi_i$, each of the variables $\dot x_{i_1}$ to $\dot x_{i_k}$ is used 
  \emph{at most once outside of weak inequalities.} The reason for
  this is that rules \eqref{eq-or1} and~\eqref{eq-or2} ensure that
  there are no disjunctions inside the $\psi_i$ that involve a weak
  variable~$\dot x$. Thus, the requirement ``in any subformula of
  $\psi_i$ of the form $\alpha \land \beta$ only $\alpha$ or $\beta$
  -- but not both -- may use $\dot x$ in a literal that is not a weak
  inequality'' from the definition of weak variables just boils down
  to ``$\dot x$ may only be used once in $\psi_i$ in a literal that is
  not a weak inequality.''

  \subparagraph*{Syntactic Transformations IV: Completing Weak Inequalities.}

  The last step before we can apply the color coding method is to
  ``complete'' the conjunctions of weak inequalities. After all the
  previous transformations have been applied, each block of weak
  quantifiers has now the form $\exists \dot x_1 \cdots \exists
  \dot x_k \bigl( \bigwedge_i \lambda_i \land
  \alpha\bigr)$ where the $\lambda_i$ are all weak inequalities
  (between some or all pairs of $\dot x_1$ to $\dot x_k$) and
  $\alpha$ contains no weak inequalities involving the~$\dot x_i$ (but
  may, of course, contain weak equalities involving the $\dot
  x_i$). Actually, the weak variables need not be $\dot x_1$ to $\dot
  x_k$, but let us assume this to keep to notation simple.

  The formula $\bigwedge_i \lambda_i$ expresses that some of the
  variables $\dot x_i$ must be different. If the formula encompasses
  all possible weak inequalities between distinct $\dot x_i$ and $\dot
  x_j$, then the formula would require that all $\dot x_i$ must be
  distinct -- exactly the situation in which color coding can be
  applied. However, some weak inequalities may be ``missing'' such as
  in the formula $\dot x_1 \neq \dot x_2 \land \dot x_2 \neq \dot x_3
  \land \dot x_1 \neq \dot x_3 \land \dot x_3 \neq \dot x_4$: This
  formula requires that $\dot x_1$ to $\dot x_3$ must be distinct and
  that $\dot x_4$ must be different from $\dot x_3$ -- but it would be
  allowed that $\dot x_4$ equals $\dot x_1$ or $\dot x_2$. Indeed, it
  might be the case that the only way to make $\alpha$ true is to make
  $\dot x_1$ equal to~$\dot x_4$. This leads to a problem in the
  context of color coding: We want to color $\dot x_1$, $\dot x_2$,
  and $\dot x_3$ differently, using, say, red, green, and blue. In
  order to ensure $\dot x_3 \neq \dot x_4$, we must give $\dot x_4$ a
  color different from blue. However, it would be wrong to color it
  red or green or using a new color like yellow since each would rule
  out $\dot x_4$ being equal or different from either $\dot x_1$ or
  $\dot x_2$ -- and each possibility must be considered to ensure that
  we miss no assignment that makes $\alpha$ true.

  The trick at this point is to reduce the problem of missing weak
  inequalities to the situation where all weak inequalities are
  present by using a large disjunction over all possible ways to unify
  weak variables without violating the weak inequalities.

  In detail, let us call a partition $P_1 \mathbin{\dot\cup} \cdots
  \mathbin{\dot\cup} P_l$ of the set $\{\dot x_1,\dots, \dot x_k\}$
  \emph{allowed by the~$\lambda_i$} if the 
  following holds: For each $P_j$ and any two different $\dot x_p,\dot
  x_q \in P_j$ none of the $\lambda_i$ is the inequality $\dot x_p
  \neq \dot x_q$. In other words, the $\lambda_i$ do not \emph{forbid}
  that the elements of any $P_j$ are identical. Clearly, the partition
  with $P_j = \{\dot x_j\}$ is always allowed by any $\lambda_i$, but
  in the earlier example, the partition $P_1 = \{\dot 
  x_1, \dot x_4\}, P_2 = \{\dot x_2\}, P_3 = \{\dot x_3\}$ would be
  allowed, while $P_1 = \{\dot x_1\}, P_2 = \{\dot x_2\}, P_3 = \{\dot
  x_3, \dot x_4\}$ would not be.

  We introduce the following notation: For a partition $P_1
  \mathbin{\dot\cup} \cdots \mathbin{\dot\cup} P_l = \{\dot x_1,\dots,
  \dot x_k\}$ we will write $\operatorname{distinct}(P_1,\dots,P_l)$ for
  $\bigwedge_{1 \le i < j \le l, \dot x_p \in P_i, \dot x_q \in P_j}
  \dot x_p \neq \dot x_q$. We claim the following:
  \begin{claim*}
    For any weak inequalities $\lambda_i$ we have 
    \begin{align*}
      \textstyle\bigwedge_i \lambda_i \equiv
      \bigvee_{\text{$P_1 \mathbin{\dot\cup} \cdots \mathbin{\dot\cup}P_l$ is allowed by the $\lambda_i$}}
      \operatorname{distinct}(P_1,\dots,P_l). 
    \end{align*}
  \end{claim*}
  
  \begin{proof}
    For the implication from left to right, assume that $\mathcal  A
    \models \bigwedge_i \lambda_i (a_1,\dots,a_k)$ for some (not
    necessarily distinct) $a_1,\dots,a_k \in |\mathcal A|$. The
    elements induce a natural partition $P_1 
    \mathbin{\dot\cup} \cdots \mathbin{\dot\cup} P_l = \{\dot
    x_1,\dots, \dot x_k\}$ where two variables $\dot x_p$ and $\dot
    x_q$ are in the same set $P_j$ if, and only if, $a_p = a_q$. Then,
    clearly, for all $i$ and $j$ with $1 \le i < j \le l$ and any
    $\dot x_p\in P_i$ and $\dot x_q \in P_j$ we have $a_i \neq
    a_j$. Thus, all inequalities in
    $\operatorname{distinct}(P_1,\dots,P_l)$ are satisfied and, hence,
    the right-hand side.

    For the other direction, suppose that $\mathcal A$ is a model of
    the right hand side for some $a_1$ to~$a_k$. Then there must be
    a partition $P_1 \mathbin{\dot\cup} \cdots \mathbin{\dot\cup} P_l$
    that is allowed by the $\lambda_i$ such that $\mathcal A$ is also
    a model of $\operatorname{distinct}(P_1,\dots,P_l)$. Furthermore,
    each $\lambda_i$ is actually present in this last formula: If
    $\dot x_p \neq \dot x_q$ is one of the $\lambda_i$, then by the
    very definition of ``$P_1 \mathbin{\dot\cup} \cdots
    \mathbin{\dot\cup} P_l$ is allowed for the $\lambda_i$'' we must
    have that $\dot x_p$ and $\dot x_q$ lie in different $P_i$ and
    $P_j$ -- which, in turn, implies that $\dot x_p \neq \dot x_q$ is
    present in  $\operatorname{distinct}(P_1,\dots,P_l)$.
  \end{proof}

  Applied to the example $\dot x_1 \neq \dot x_2 \land \dot x_2 \neq
  \dot x_3 \land \dot x_1 \neq \dot x_3 \land \dot x_3 \neq \dot x_4$
  from above, the claim states the following: Since there are three
  partitions that are allowed by these literals (namely the one in
  which each variable gets its own equivalence class, the one where
  $\dot x_1$ and $\dot x_4$ are put into one class, and  the one where
  $\dot x_2$ and $\dot x_4$ are put into one class), this formula is
  equivalent to: $
      \operatorname{distinct}(\{\dot x_1\}, \{\dot x_2\}, \{\dot x_3\}, \{\dot x_4\}) \lor
      \operatorname{distinct}(\{\dot x_1,\dot x_4\}, \{\dot x_2\}, \{\dot x_3\}) \lor
      \operatorname{distinct}(\{\dot x_1\}, \{\dot x_2,\dot x_4\}, \{\dot x_3\})$.

  The claim has the following trivial corollary:
  \begin{corollary*}
    For any weak inequalities $\lambda_i$ involving only variables
    from $\{\dot x_1, \dots, \dot x_k\}$ we have\\ $\textstyle \exists \dot x_1 \cdots \exists \dot x_k
      \bigl(\bigwedge_i \lambda_i \land \alpha) 
      \equiv\textstyle
\bigvee_{\text{$P_1 \mathbin{\dot\cup} \cdots \mathbin{\dot\cup}P_l$ is allowed by the $\lambda_i$}}
      \exists \dot x_1 \cdots \exists \dot x_k
                 (\operatorname{distinct}(P_1,\dots,P_l) \land \alpha)$. 
  \end{corollary*}

  As in the previous transformations we now apply the equivalence from
  the corollary from left to right. If we create copies of $\alpha$
  during this process, we rename the weak variables in these copies to
  ensure, once more, that each weak variable is unique. In our example
  formula $\phi$, there is only one place where the transformation 
  changes anything: The middle weak quantifier block (the
  $\exists \dot x_3 \exists \dot x_4$ block). For the first and the
  last block, the literals $\dot x_1 \neq \dot x_2$ and $\dot x_5 \neq
  \dot x_6$, respectively, already rule out all partitions except for
  the trivial one. For the middle block, however, there are \emph{no}
  weak inequalities at all and, hence, there are now two allowed
  partitions: First, $P_1 = \{\dot x_3\}, P_2 = \{\dot x_4\}$, but
  also $P_1 = \{\dot x_3,\dot x_4\}$. This means that we get a copy of
  the middle block where $\dot x_3$ and $\dot x_4$ are required to be
  different -- and we renumber them to $\dot x_7$ and $\dot x_8$:
  \begin{align*}
    \phi\equiv\forall v \exists \dot x_1\exists \dot x_2 \bigl(  
    &\dot x_2 \neq
      \dot x_1 \land 
    \exists a \exists v_1 (v_1 = \dot x_1 \land E a v_1) \land {}\\
    &\forall c \bigl(\exists \dot x_3 \exists \dot x_4 (\phantom{\dot x_3 \neq
      \dot x_4 \land{}} \exists v_1(v_1 =
      \dot x_3 \land \exists v_2(v_2 = \dot x_4 \land E v_1 v_2))) \lor{}\\
    &\phantom{\forall c \bigl(}\exists \dot x_7 \exists \dot x_8 (
      {\dot x_7 \neq
      \dot x_8 \land{}} \exists v_1(v_1 =
      \dot x_7 \land \exists v_2(v_2 = \dot x_8 \land E v_1 v_2))) \lor{}\\
    &\phantom{\forall c\bigl(} \exists \dot x_5 \exists \dot x_6(\dot x_5 \neq \dot x_6 \land \exists z(Pz\land Qc))\bigr)\bigr).
  \end{align*}

  \subparagraph*{Applying Color Coding.}

  We are now ready to apply the color coding technique; more
  precisely, to repeatedly apply Theorem~\ref{thm-cc-intro} to the
  formula~$\phi$. Before we do so, let us summarize the structure of
  $\phi$:
  \begin{enumerate}
  \item All weak quantifiers come in blocks, and each such block
    either directly follows a universal quantifier or follows a
    disjunction after a universal quantifier. In particular, on any
    root-to-leaf path in the syntax tree of $\phi$ between any two
    blocks of weak quantifiers there is at least one universal
    quantifier.
  \item All blocks of weak quantifiers have the form
    \begin{align}
      \exists \dot x_{i_1}
      \cdots \exists \dot x_{i_k} \bigl(\operatorname{distinct}(P_1,\dots,P_l) \land \alpha\bigr)\label{eq-form}
    \end{align}
    for some partition $P_1 \mathbin{\dot\cup} \cdots
    \mathbin{\dot\cup} P_l = \{x_{i_1}, \dots, x_{i_k}\}$ and for some
    $\alpha$ in which the only literals that contain any $\dot
    x_{i_j}$ are of the form $\dot x_{i_j} = 
    y$ for a strong variable~$y$ that is bound by an existential quantifier
    inside~$\alpha$. Furthermore, none of these weak equality literals
    is in the scope of a universal quantifier inside~$\alpha$. (Of
    course, all variables in $\phi$ are in the scope of a universal
    quantifier since we added one at the start, but the point is that
    none of the $\dot x_i$ is in the scope of a universal quantifier
    that is inside~$\alpha$.)
  \end{enumerate}
  In $\phi$ there may be several blocks of weak quantifiers, but at
  least one of them (let us call it~$\beta$) must have the form
  \eqref{eq-form} where $\alpha$ contains no weak variables other than
  $\dot x_{i_1}$ to $\dot x_{i_k}$. (For instance, in our example formula,
  this is the case for the blocks 
  starting with $\exists \dot x_3 \exists \dot x_4$, for $\exists \dot
  x_7 \exists \dot x_8$, and for $\exists \dot x_5 \exists \dot x_6$, but not for
  $\exists \dot x_1 \exists \dot x_2$ since, here, the corresponding
  $\alpha$ contains all of the rest of the formula.)
  In our example, we could choose $\beta = \exists \dot x_7 \exists
  \dot x_8 (\dot x_7 \neq \dot x_8 \land \exists v_1(v_1 = 
  \dot x_7 \land \exists v_2(v_2 = \dot x_8 \land E v_1 v_2)))$ and
  would then have
  \begin{align*}
    \alpha = \exists v_1(v_1 =
    \dot x_7 \land \exists v_2(v_2 = \dot x_8 \land E v_1 v_2)).
  \end{align*}
  
  We build a new formula $\alpha'$ from $\alpha$ as follows: We replace
  each occurrence of a weak equality $\dot x_i = y$ in $\alpha$ for some
  weak variable $\dot x_i \in P_j$ and some strong variable~$y$ by the 
  formula~$C_j y$. In our example, where $P_1 = \{\dot x_7\}$ and $P_2
  = \{\dot x_8\}$ we would get
  \begin{align*}
    \alpha' = \exists v_1(C_1 v_1 \land
    \exists v_2(C_2 v_2 \land E v_1 v_2)).
  \end{align*}
  An important observation at
  this point is that $\alpha'$ \emph{contains 
    no weak variables any longer, while no additional variables have
    been added.} In particular, the quantifier rank of $\alpha'$
  equals the \emph{strong} quantifier rank of $\alpha$ and the number
  of variables in $\alpha'$ equals the number of \emph{strong}
  variables in~$\alpha$.

  Note that the literals $C_j y$ and also $\dot x_i =y$ are positive
  since the formulas are in negation normal form. Hence, they have the
  following monotonicity 
  property: If some structure together with some assignment to the
  free variables is a model of $\alpha$ or $\alpha'$, but a literal
  $\dot x_i = y$ or $C_j y$ is false, the structure will still be a
  model if we replace the literal by a tautology. 

  For simplicity, in the following, we assume that $\dot x_{i_1}$ to
  $\dot x_{i_k}$ are just $\dot x_1$ to $\dot x_k$. Also for
  simplicity we assume that $\beta$ contains no free variables when,
  in fact, it can. However, these variables cannot be any of the
  variables~$y$ for which we make changes and, thus, it keeps the
  notation simpler to ignore the additional free variables here. The
  following statement simply holds for all assignments to them:

  \begin{claim*}
    Let $P_1 \mathbin{\dot\cup} \cdots \mathbin{\dot\cup} P_l = \{x_1,
    \dots, x_k\}$. Then for each structure $\mathcal A$, the following
    are equivalent: 
    \begin{enumerate}
    \item $\mathcal A \models \exists \dot x_1 \cdots \exists \dot x_k
      \bigl(\operatorname{distinct}(P_1,\dots,P_l) \land \alpha\bigr)$. 
    \item There are elements $a_1, \dots, a_k \in |\mathcal A|$ with 
      $\mathcal A \models \alpha (a_1,\dots,a_k)$ and such that $a_p
      \neq a_q$ whenever $\dot x_p \in P_i$, $\dot x_q \in P_j$, and
      $i\neq j$. 
    \item There is an $l$-coloring $\mathcal B$ of $\mathcal A$ such
      that $\mathcal B \models \alpha'$.
    \end{enumerate}
  \end{claim*}
  
  \begin{proof}
    For the proof of the claim, it will be useful to apply some
    syntactic transformations to $\alpha$ and~$\alpha'$. Just like the
    many transformations we encountered earlier, these transformations
    yield equivalent formulas and, thus, it suffices to prove the
    claim for them (since the claim is about the models of $\alpha$
    and $\alpha'$). However, these transformation are needed only to
    prove the claim, they are \emph{not} part of the ``chain of
    transformations'' that is applied to the original formula (they
    increase the number of strong variables far too much).

    In $\alpha$ there will be some occurrences of literals of the form
    $\dot x_i =y$. For each such occurrence, there will be exactly one
    subformula in $\alpha$ of the form $\exists y(\gamma)$ where
    $\gamma$ contains $\dot x_i =y$. We now apply two syntactic
    transformations: First, we replace $y$ in $\exists y(\gamma)$ by a
    fresh new variable~$y_i$ (that is, we replace all free occurrences
    of~$y$ inside $\gamma$ by~$y_i$ and we replace the leading
    $\exists y$ by~$\exists y_i$). Second, we ``move all $\exists y_i$
    to the front'' 
    by simply deleting all occurrences of $\exists y_i$ from $\alpha$,
    resulting in a formula $\delta$, and then adding the block
    $\exists y_1 \cdots \exists y_k$ before $\delta$. As an
    example, if we apply these transformations to $\alpha = \exists 
    v_1(\dot x_7 = v_1 \land \exists v_2(\dot x_8 = v_2 \land E v_1
    v_2))$, the first transformation yields $\exists y_7(\dot x_7 =
    y_7 \land \exists y_8(\dot x_8 = y_8 \land E y_7 y_8))$ and the
    second one yield the new
    \begin{align*}
      \alpha = \exists y_1 \cdots \exists y_8
      (\dot x_7 = y_7 \land \dot x_8 = y_8 \land E y_7 y_8).
    \end{align*}
    
    In $\alpha'$, we apply exactly the same transformations, only now
    the literals we look for are not $\dot x_i = y$, but $C_j y$. We
    still apply the same renaming of~$y$ (namely to $y_i$ and not to
    $y_j$) as in $\alpha$ and apply the same movement of the
    quantifiers. This results in a new formula $\alpha'$
    of the form $\exists y_1 \cdots \exists y_k(\delta')$. For
    $\alpha' =  \exists  v_1(C_1 v_1 \land \exists v_2(C_2 v_2 \land E
    v_1 v_2))$ we get the new
    \begin{align*}
      \alpha' = \exists y_1 \cdots \exists y_8 (C_1
      y_7 \land C_2 y_8 \land E y_7 y_8)
    \end{align*}
    and $\delta'$ is now the inner part without the quantifiers. 

    Let us now prove the claim. The first two items are trivially
    equivalent by the 
    definition of $\operatorname{distinct}(P_1,\penalty100\dots,P_l)$.

    The second statement implies the third: To show this, for $j \in
    \{1,\dots,l\}$ we first set  
    $C_j^{\mathcal B} = \{a_i \mid \dot x_i \in P_j\}$ and then add
    $|\mathcal A| \setminus \{a_1,\dots,a_k\}$ to, say, $C_1^{\mathcal
      B}$ in order to create a correct partition. This setting clearly
    ensures that whenever $\dot x_i = y$ holds in $\alpha$, we also have
    $C_j y$ holding in~$\alpha'$. Since $\alpha'$ differs from
    $\alpha$ only on the literals of the form $\dot x_i = y$ (which
    got replaced by $C_j y$), since we just saw that when $\dot x_i =
    y$ holds in $\alpha$, the replacements $C_j y$ holds in $\alpha'$,
    and since $\alpha$ has the monotonicity property (by which it does
    matter when \emph{more} literals of the form $C_i y$ hold in
    $\alpha'$ than did in~$\alpha$), we get the third statement.

    The third statement implies the second: Let an $l$-coloring
    $\mathcal B$ of $\mathcal A$ be given with $\mathcal B \models
    \alpha'$. Since $\alpha' = \exists y_1 \cdots \exists y_k
    (\delta')$, there must now be elements $b_1,\dots,b_k \in
    |\mathcal A|$ such that $\mathcal B \models
    \delta'(b_1,\dots,b_k)$. We define new elements $a_i \in |\mathcal
    A|$ as follows: If $b_i \in C_i^{\mathcal B}$, let
    $a_i=b_i$. Otherwise, let $a_i$ be an arbitrary element of
    $C_i^{\mathcal B}$.  We show in the following that the $a_i$
    constructed in this way can be used in the  second statement, that
    is, we claim that $\mathcal A \models \alpha (a_1,\dots,a_k)$ and
    the $a_i$ have the distinctness property from the claim.

    First, recall that $\alpha$ is of the form $\exists y_1
    \cdots \exists y_k (\delta)$ (because of
    the syntactic transformations we applied for the purposes of the
    proof of this claim) and $\delta$ contains literals of the form
    $\dot x_i = y_i$, where the $\dot x_i$ are the free variables for
    which the values $a_i$ and now plugged in. We claim that $\mathcal
    A \models \delta(a_1,\dots,a_k,b_1,\dots,b_k)$, that is, we claim
    that if we plug in $a_1$ to $a_k$ for the free variables $\dot
    x_1$ to $\dot x_k$ in $\delta$ and we plug in $b_1$ to $b_k$
    for the (additional) free variables $y_1$ to $y_k$ in
    $\delta$, then $\delta$ holds in~$\mathcal A$. To see this, recall
    that $\mathcal B \models \delta'(b_1,\dots,b_k)$ holds and
    $\delta'$ is identical to $\delta$ except that $\dot x_i = y_i$
    got replaced by $C_j y_i$. In particular, by construction of
    the~$a_i$, whenever $C_j y_i$ holds in $\mathcal B$ with $y_i$
    being set to $b_i$ (that is, whenever $b_i \in C_j^{\mathcal B}$),
    we clearly also have that $\dot x_i = y_i$ holds in $\mathcal A$
    with $\dot x_i$ being set to $a_i$ and $y_i$ being set to~$b_i$
    (since we let $a_i = b_i$ whenever $b_i \in C_j^{\mathcal
      B}$). But, then, by the monotonicity property, we know that
    $\mathcal A \models \delta(a_1,\dots,a_k,b_1,\dots,b_k)$  will
    hold.

    Second, we argue that the distinctness property holds, that is,
    $a_p \neq a_q$ whenever $\dot x_p \in P_i$, $\dot x_q \in P_j$, and
    $i\neq j$. However, our construction ensured that we always have
    $a_r \in C_s^{\mathcal B}$ for the~$s$ with $\dot x_r \in P_s$. In
    particular, $\dot x_p \in P_i$ and $\dot x_q \in P_j$ for
    $i\neq j$ implies that $a_p$ and $a_q$ lie in two different color
    classes and are, hence, distinct.
  \end{proof}

  By the claim, $\mathcal A \models \beta$ is equivalent to there
  being an $l$-coloring $\mathcal B$ of $\mathcal A$ such
  that $\mathcal B \models \alpha'$. We now apply
  Theorem~\ref{thm-cc-intro} to~$\alpha'$ (as $\phi$), which yields a
  new formula $\alpha''$ (called $\phi'$ in the theorem) with the
  property $\mathcal A \models \alpha'' \iff \mathcal A \models
  \beta$. The interesting thing about $\alpha''$ is, of course, that
  it has the same quantifier rank and the same number of variables as
  $\alpha'$ plus some constant. Most importantly, we already pointed
  out earlier that $\alpha'$ does \emph{not} contain any weak
  variables and, hence, the quantifier rank of $\alpha''$ is the same
  as the \emph{strong} quantifier rank of $\beta$ and the number of
  variables in $\alpha''$ is the same as the number of \emph{strong}
  variables in $\beta$ -- plus some constant.

  Applying this transformation to our running example $\phi$ and
  choosing as $\beta$ once more the subformula starting with $\exists
  \dot x_7 \exists \dot x_8$, we would get the following formula
  (ignoring the technical issues how, exactly, the hashing is
  implemented, see the proof of Theorem~\ref{thm-cc-intro} for the
  details):    
  \begin{align*}
    \forall v \exists \dot x_1\exists \dot x_2 \bigl(  
    &\dot x_2 \neq
      \dot x_1 \land 
    \exists a \exists v_1 (v_1 = \dot x_1 \land E a v_1) \land {}\\
    &\forall c \bigl(\exists \dot x_3 \exists \dot x_4 (\phantom{\dot x_3 \neq
      \dot x_4 \land{}} \exists v_1(v_1 =
      \dot x_3 \land \exists v_2(v_2 = \dot x_4 \land E v_1 v_2))) \lor{}\\
    &\phantom{\forall c \bigl(}\textstyle \bigvee_g \exists p \exists q 
      \exists v_1(\op{hash}_{g}(v_1,p,q)=1 \land \exists v_2(\op{hash}_{g}(v_2,p,q)=2 \land E v_1 v_2)) \lor{}\\
    &\phantom{\forall c\bigl(} \exists \dot x_5 \exists \dot x_6(\dot x_5 \neq \dot x_6 \land \exists z(Pz\land Qc))\bigr)\bigr).
  \end{align*}
  
  We can now repeat the transformation to replace each block $\beta$
  in this way. Observe that in each transformation we can reuse the
  variables (in particular, $p$ and $q$) introduced by the color
  coding:
  \begin{align*}
    \forall v  \textstyle\bigvee_g \exists p \exists q \bigl(  
    &\exists a \exists v_1 (\op{hash}_g(v_1,p,q) = 1 \land E a v_1) \land {}\\
    &{\forall c \bigl(}\textstyle \bigvee_g \exists p \exists q 
      \exists v_1(\op{hash}_{g}(v_1,p,q)=1 \land \exists v_2(\op{hash}_{g}(v_2,p,q)=1 \land E v_1 v_2)) \lor{}\\
    &\phantom{\forall c \bigl(}\textstyle \bigvee_g \exists p \exists q 
      \exists v_1(\op{hash}_{g}(v_1,p,q)=1 \land \exists v_2(\op{hash}_{g}(v_2,p,q)=2 \land E v_1 v_2)) \lor{}\\
    &\textstyle\phantom{\forall c\bigl(}   \bigvee_g \exists p
      \exists q \exists z(Pz\land Qc)\bigr)\bigr). 
  \end{align*}

  In conclusion, we see that we can transform the original formula
  $\phi$ to a new formula $\phi'$ with the following properties:
  \begin{itemize}
  \item We added new variables and quantifiers to $\phi'$ compared
    to $\phi$ during the first transformation steps, but the number we
    added depended only on the signature~$\tau$ (it was three
    times the maximum arity of relations in~$\tau$).
  \item We then removed all weak variables from $\phi$ in~$\phi'$.
  \item We added some variables to $\phi'$ each time
    we applied Theorem~\ref{thm-cc-intro} to a block~$\beta$. The
    number of variables we added is constant since
    Theorem~\ref{thm-cc-intro} adds only a constant number of
    variables and since we can always reuse the same set of variables
    each time the theorem is applied.
  \item We also added some quantifiers to $\phi'$ each time we applied
    Theorem~\ref{thm-cc-intro}, which increases the quantifier rank of
    $\phi'$ compared to $\phi$ by more than a constant. However, the
    essential quantifiers we add are $\exists p \exists q$ and these
    are \emph{always added directly after a universal quantifier or
      directly after a disjunction after a universal quantifier.} Since the
    strong quantifier rank of~$\phi$ is at least the quantifier rank
    of $\phi$ where we only consider the universal quantifiers (the
    ``universal quantifier rank''), the two added nested quantifiers
    per universal quantifiers can add to the quantifier rank of
    $\phi'$ at most twice the universal quantifier rank.
  \end{itemize}

  Putting it all together, we see that $\phi'$ is equivalent to
  $\phi$, that $\phi'$ has a quantifier rank that is at most $3
  \operatorname{strong-qr}(\phi) + O(\operatorname{arity}(\tau))$, and the $\phi'$ contains at
  most $\operatorname{strong-bound}(\phi) + O(\operatorname{arity}(\tau))$ variables.  
\end{proof}

We already mentioned that the notion of weak existential quantifiers
begs a dual: By Theorem~\ref{theorem-strong-character}, for $\phi =
\exists \dot x_1 \cdots \exists \dot x_k (\psi)$ there is an
equivalent formula $\phi'$ with $\operatorname{qr}(\phi') =
O(\operatorname{strong-qr}(\phi))$. Since, trivially,
$\operatorname{qr}(\neg \phi') = \operatorname{qr}(\phi')$, the
formula $\neg \phi$ is also equivalent to a formula of quantifier rank
$O(\operatorname{strong-qr}(\phi))$. The normal form of $\neg
\phi$ starts with $\forall x_1 \cdots \forall x_k$ to which
Theorem~\ref{theorem-strong-character} does not apply ``at all'' -- but
the dual of the theorem applies, where we call the leading quantifier
in a (sub)formula $\forall x (\phi)$ \emph{weak} if no
\emph{existential} binding inside~$\phi$ depends on~$x$ and in all
subformulas of~$\phi$ of the form $\alpha \lor \beta$ at most one of
$\alpha$ and $\beta$ may contain a literal that contains~$x$ and is
not of the form $x = y$ (note that this is now an equality). More
interestingly, we can even show that both kinds of weak quantifiers
may be present: 

\begin{theorem}\label{theorem-strong-character-ext}
  Theorem~\ref{theorem-strong-character} still holds when $\phi$ may
  contain both existential and universal weak variables, none of which
  count towards the strong quantifier rank nor count as strong bound
  variables. 
\end{theorem}
\begin{proof}
  Given a formula $\phi$ that contains both existential and universal
  weak quantifiers, we apply a syntactic preprocessing that
  ``separates these quantifiers and moves them before their dual
  strong quantifiers.'' The key observation that 
  makes these transformations possible in the mixed case is that weak
  existential and weak universal quantifiers commute: For instance,
  $\exists \dot x   (\alpha \land \forall \dot y (\beta)) \equiv
  \forall \dot y (\beta \land \exists \dot x (\alpha))$ since $\dot x$
  and $\dot y$ cannot depend on one another by the core property of
  weak quantifiers ($\alpha$ cannot contain~$\dot y$ and $\beta$
  cannot contain~$\dot x$). Once we have sufficiently separated the
  quantifiers, we can repeatedly apply
  Theorem~\ref{theorem-strong-character} or its dual to each block
  individually. 

  As a running example, let us use the following formula $\phi$:
  \begin{align*}
    \exists \dot x\exists a (E \dot x a \land \forall b(Eba \lor Eab) \land \forall
    \dot y(E a \dot y \land \exists \dot z(E a\dot z)) \land \exists
    \dot w(E \dot w a)),
  \end{align*}
  which mixes existential and universal weak variables rather freely.

  Similar to the proof of Theorem~\ref{theorem-strong-character}, for 
  technical reasons we first add the superfluous quantifiers $\exists
  v \forall v$ for a fresh strong variable~$v$ at the beginning
  of the formula. 

  Our main objective is to get rid of alternations of weak
  universal and weak existential quantifiers without a strong
  quantifier in between. In the example, this is the case, for
  instance, for $\exists \dot x(\dots \forall \dot y(\dots \exists
  \dot z\dots))$. We get rid of these situations by pushing all
  quantifiers (weak or strong) \emph{down} as far as possible (later
  on, when we apply Theorem~\ref{theorem-strong-character}, we will
  push them up once more). Let us write $\mathring x$ to indicate that
  $x$ may both be a weak or a strong variable.

  If $\beta$ does not contain~$\mathring x$ as a free variable, we
  can apply the following equivalences from left to right (and, of
  course, commutatively equivalent ones):
  \begin{align}
    \exists \mathring x(\alpha \land \beta) &\equiv \exists \mathring x(\alpha) \land \beta, \label{eq-tr1}\\
    \exists \mathring x(\alpha \lor \beta) &\equiv \exists \mathring x(\alpha) \lor \beta, \label{eq-tr2}\\
    \forall \mathring x(\alpha \land \beta) &\equiv \forall \mathring x(\alpha) \land \beta, \label{eq-tr3}\\
    \forall \mathring x(\alpha \lor \beta) &\equiv \forall \mathring x(\alpha) \lor \beta. \label{eq-tr4}
  \end{align}
  Note that the definition of weak variables forbids that a
  universally bound variable depends on an existential weak variable
  (and vice versa). This means that in the first two lines, if
  $\mathring x$ is actually the weak variable $\dot x$ and if $\beta$
  starts with $\forall \dot y$, we can automatically apply both
  equivalences. Similarly, if in the last two lines $\beta$ starts with
  $\exists \dot y$, we can also apply both equivalences. 

  Furthermore, we also apply the following general equivalences as
  long as possible:
  \begin{align}
    \exists \mathring x(\alpha \land (\beta \lor \gamma))
    &\equiv
    \exists \mathring x((\alpha \land \beta) \lor (\alpha \land
      \gamma)), \label{eq-tr5} \\
    \exists \mathring x(\alpha \lor \beta) &\equiv \exists \mathring x(\alpha) \lor \exists \mathring x(\beta), \\
    \forall \mathring x(\alpha \lor (\beta \land \gamma))
    &\equiv
    \forall \mathring x((\alpha \lor \beta) \land (\alpha \lor \gamma)), \\
    \forall \mathring x(\alpha \land \beta) &\equiv \forall \mathring x(\alpha) \land \forall \mathring x(\beta).\label{eq-tr-last}
  \end{align}  
  Applied to our example, we would get:
  \begin{align*}
    \exists v \forall v\,
    \exists \dot x\exists a (E \dot x a \land \forall b(Eba \lor Eab) \land \forall
    \dot y(E a \dot y) \land \exists \dot z(E a\dot z) \land \exists
    \dot w(E \dot w a)).
  \end{align*}

  As a final transformation, we ``sort'' the operands of disjunctions
  and conjunctions: We replace a subformula $\alpha \land \beta$ in
  $\phi$ by $\beta \land \alpha$ and we replace $\alpha \lor \beta$ by
  $\beta \lor \alpha$, whenever $\beta$ contains no weak universal
  variables, but $\alpha$ does, and also whenever $\alpha$ contains no
  weak existential variables, by $\beta$ does. For our example, this
  means that we get the following:
  \begin{align*}
    \exists v \forall v\,
    \exists \dot x\exists a (
    \exists \dot z(E a\dot z)
    \land \exists \dot w(E \dot w a)
    \land E \dot x a
    \land \forall b(Eba \lor Eab)
    \land \forall \dot y(E a \dot y)).
  \end{align*}

  The purpose of the transformations was to achieve the situation
  described in the next claim:
  \begin{claim*}
    Assume that the above transformations have been applied
    exhaustively to~$\phi$ and assume $\phi$ contains both existential
    and universal weak variables. Consider the maximal
    subformulas $\alpha_i$ of~$\phi$ that contain no weak universal variables
    and the maximal subformulas $\beta_i$ of~$\phi$ that contain no
    weak existential variables. Then for some $i$ and some~$\gamma$
    one of the following formulas is a subformula of~$\phi$: $\forall
    x(\alpha_i \lor \gamma)$ or $\exists x(\gamma \land \beta_i)$.
  \end{claim*}

  In our example, there is only a single maximal~$\alpha_1$, namely $
  \exists \dot z(E a\dot z)
  \land \exists \dot w(E \dot w a)
  \land E \dot x a
  \land \forall b(Eba \lor Eab)$, and a single maximal
  $\beta_1$, namely $ \forall b(Eba \lor Eab) \land \forall
  \dot y(E a \dot y)$. The claim holds since $\exists a(\gamma \land
  \beta_1)$ is a subformula for $\gamma =\exists
  \dot z(E a\dot z) \land \exists \dot w(E \dot w a) \land E \dot x a$.
  
  \begin{proof}
    Consider any $\alpha$ among the $\alpha_i$. Since $\alpha$ is
    maximal but not all of~$\phi$, there must be a $\beta$ among the
    $\beta_i$ such that either $\alpha \lor \beta$ or $\alpha \land
    \beta$ is also a subformula of~$\phi$. Let us call it $\delta$ and
    consider the minimal subformula $\eta$ of $\phi$ that contains
    $\delta$ and starts with a quantifier.

    This quantifier cannot be a weak quantifier: Suppose
    it is $\exists \dot x$ (the case $\forall \dot x$ is perfectly
    symmetric). Since we can no longer apply one of the equivalences 
    \eqref{eq-tr1} to~\eqref{eq-tr-last}, the formula $\eta$ must have
    the form $\exists \dot x \bigwedge_i \psi_i$ (where the $\psi_i$
    are not of the form $\rho\land\sigma$) such that all
    $\psi_i$ contain~$\dot x$ (otherwise \eqref{eq-tr1} would be
    applicable) and such that none of the $\psi_i$ is of the form
    $\rho \lor \sigma$ (otherwise \eqref{eq-tr5} would be
    applicable). This implies that all $\psi_i$ start with a
    quantifier. Since $\eta$ was minimal to contain~$\delta$, we
    conclude that one $\psi_i$ must be $\alpha$ and another one must
    be~$\beta$. But, then, $\beta$ contains a weak existential
    variable, namely $\dot x$, which we ruled out.

    Since $\eta$ does not start with a weak quantifier, it must start
    with a strong quantifier. If it is $\exists x$, by the same
    argument as before we get that $\eta$ must have the form $\exists
    x \bigwedge_i \psi_i$ with some $\psi_i$ equal to $\alpha$ and
    some other $\psi_j$ equal to~$\beta$. But, then, we have found the
    desired subformula of $\phi$ if we set $\gamma$ to
    $\bigwedge_{i\neq j} \psi_i$. If the strong quantifier is $\forall
    x$, a perfectly symmetric argument shows that $\eta$ must have the
    form $\forall x \bigvee_i \psi_i$ with some $\psi_j = \alpha$,
    which implies the claim for $\gamma = \bigvee_{i\neq j} \psi_i$.
  \end{proof}

  The importance of the claim for our argument is the following: As
  long as $\phi$ still contains both existential and universal weak
  variables, we still find a subformula $\alpha$ or $\beta$ that
  contains only existential or universal weak variables such that if
  we go up from this subformula in the syntax tree of~$\phi$, the next
  quantifier we meet is a strong quantifier. This means that we can 
  now apply Theorem~\ref{theorem-strong-character} or its dual to this
  subformula, getting an equivalent new formula $\alpha'$ or $\beta'$
  whose quantifier rank equals the strong quantifier rank of $\alpha$
  or~$\beta$, respectively, times a constant factor. Furthermore,
  similar to the argument at the end of the proof of
  Theorem~\ref{theorem-strong-character} where we processed
  one~$\beta$ after another, each time a replacement takes place,
  there is a strong quantifier that contributes to the strong
  quantifier rank of~$\phi$. 
\end{proof}

\section{Syntactic Proofs and Natural Problems}\label{section-applications}

The special allure of descriptive complexity theory lies in the
possibility of proving that a problem has a certain complexity just by
describing the problem in the right way. The ``right way'' is, of
course, a logical description that has a certain syntax (such as
having a bounded strong quantifier rank). In the following we present
such descriptions for several natural problems and thereby bound their
complexity ``in a purely syntactic way.'' First, however, we present
``syntactic tools'' for describing problems more easily. These tools
are built on top of the notion of strong and weak quantifiers.

\subsection{Syntactic Tools: New Operators}

It is common in mathematical logic to distinguish between the core
syntax and additional ``shorthands'' built on top of the core
syntax. For instance, while $\neg$ and $\lor$ are typically considered
to be part of the core syntax of propositional logic, the notation $a \to b$
is often seen as a shorthand for $\neg a \lor b$. In a
similar way, we now consider the notions of weak variables and
quantifiers introduced in the previous section as our ``core syntax''
and  build a number of useful shorthands on top of them. Of course,
just as $a \to b$ has an intended semantic meaning that the expansion
$\neg a \lor b$ of the shorthand  must reflect, the shorthands we
introduce also have an intended semantic meaning, which we specify.

As a first example, consider the common notation $\exists^{\ge k} x
(\phi(x))$, whose intended semantics is ``there are at least $k$
different elements in the universe that make $\phi(x)$ true.'' While
this notation is often considered as a shorthand for $\smash{\exists x_1
\cdots \exists x_k \bigwedge_{i \neq j} x_i \neq x_j \land
\bigwedge_{i=1}^k \phi(x_i)}$ we will consider it a shorthand for
the equivalent, but slightly more complicated formula $\exists \dot x_1
\cdots \exists \dot x_k \smash{\bigwedge_{i \neq j}} \dot x_i \neq \dot x_j
\land \smash{\bigwedge_{i=1}^k} \exists x(x = \dot x_i \land \phi(x))$. The
difference is, of course, that the strong quantifier rank is now much
lower and, hence, by Theorem~\ref{theorem-strong-character} we can
replace any occurrence of $\exists^{\ge k} x (\phi(x))$ by a formula
of quantifier rank $\operatorname{qr}(\phi) + O(1)$. In all of the
following notations, $k$ and $s$ are arbitrary values. The indicated strong
quantifier rank for the notation is that of its
expansion. The \emph{semantics} describe which structures 
$\mathcal A$ are models of the formula. 

\begin{notation}
  {$\exists^{\ge k} x (\phi(x))$}
  {$1+\operatorname{strong-qr}(\phi)$}
\item[Semantics] There are $k$ distinct $a_1,\dots,a_k \in |\mathcal
  A|$ with $\mathcal A \models \phi(a_i)$ for all~$i$.
\item[Expansion] $\exists \dot x_1
\cdots \exists \dot x_k \bigwedge_{i \neq j} \dot x_i \neq \dot x_j
\land \bigwedge_{i=1}^k \exists x(x = \dot x_i \land \phi(x))$
\end{notation}

\begin{notation}
  {$\exists^{\le k} x (\phi(x))$}
  {$1+\operatorname{strong-qr}(\phi)$}
\item[Semantics] There are at most $k$ distinct $a_1,\dots,a_k \in |\mathcal
  A|$ with $\mathcal A \models \phi(a_i)$ for all~$i$.
\item[Expansion] $ \forall \dot x_1 \cdots \forall \dot x_{k+1} 
  \bigvee_{i \neq j} \dot x_i = \dot x_j \lor \bigvee_{i=1}^{k+1} \forall
  x(x\neq\dot x_i \lor \neg \phi(x))$ ($\equiv \neg \exists^{\ge k+1}x (\phi(x))$) 
\end{notation}

\begin{notation}
  {$\exists^{=k} x (\phi(x))$}
  {$1+\operatorname{strong-qr}(\phi)$}
\item[Semantics] There are exactly $k$ distinct $a_1,\dots,a_k \in |\mathcal
  A|$ with $\mathcal A \models \phi(a_i)$ for all~$i$.
\item[Expansion] $\exists^{\ge k} x (\phi(x)) \land \exists^{\le k} x (\phi(x))$
\end{notation}
%
The next notation is useful for ``binding'' a set of vertices to weak
or strong variables. The binding contains the allowed ``single use'' of 
the weak variables in the sense of Definition~\ref{def-weak}, but they
can still be used in inequality literals. Let $\mathring x$ 
indicate that $x$ may be weak or strong.
\begin{notation}
  {$\{\mathring x_1, \dots, \mathring x_k\} = \{x \mid \phi(x)\}$}
  {$1+\operatorname{strong-qr}(\phi)$}
\item[Semantics] Let $a_1,\dots,a_k \in |\mathcal A|$ be the
  assignments to the $\mathring x_i$ (note that they need \emph{not} be
  distinct). Then $\{a_1,\dots,a_k\} = \bigl\{ a \in \left|\mathcal
    A\right| \bigm| \mathcal A \models \phi(a)\bigr\}$ must hold.
\item[Expansion] $\bigwedge_{i=1}^k \exists x\bigl(x=\mathring x_i \land
  \phi(x)\bigr) \land {}$ \hfill \emph{// ensure $\{\mathring x_1, \dots, \mathring
    x_k\} \subseteq \{x \mid \phi(x)\}$}\\
  $\bigvee_{s=1}^k\bigl(\exists^{=s} x (\phi(x))
  \land{}$
  \hfill
  \emph{// bind $s$ to $|\{x \mid \phi(x)\}|$}
  \\
  \hbox{}\quad$\bigvee_{I \subseteq \{1,...,k\}, |I| = s} \bigwedge_{i,j \in I,
  i\neq j} \mathring x_i \neq \mathring x_j\bigr)$.\hfill  \emph{//
  ensure $|\{\mathring x_1, \dots, \mathring 
    x_k\}| \ge s$}
\end{notation}
The final notation can be thought of as a ``generalization of
$\exists^{=k}$'' where we not only 
ask whether there are exactly $k$ distinct $a_i$ with a property
$\phi$, but whether these $a_i$ then also have an \emph{arbitrary}
special additional property. Formally, let $Q \subseteq \struc[\tau]$
be an arbitrary $\tau$-problem. We write $\mathcal A[I]$ for the
substructure of $\mathcal A$ induced on a subset $I \subseteq |\mathcal A|$.

\begin{notation}
  {$\op{induced}^{\mathrm{size}=k}\{x \mid \phi(x)\} \in Q$}
  {$1+\operatorname{strong-qr}(\phi)+ \operatorname{arity}(\tau)$}
\item[Semantics] The set $I = \{ a \in |\mathcal A| \mid \mathcal A
  \models \phi(a)\}$ has size exactly~$k$ and $\mathcal A[I] \in Q$.
\item[Expansion] Assuming for simplicity that $\tau$ contains only $E^2$
  as non-arithmetic predicate:
  \begin{align*}
    \textstyle
    \exists^{=k} x(\phi(x)) \land \bigvee_{\mathcal A \in Q,
    |\mathcal A| = \{1,\dots,k\}} &\textstyle\bigwedge_{(i,j) \in E^{\mathcal A}}
    \exists x \exists y(\pi_i(x)\land \pi_j(y) \land E xy)
    \land{}
    \\
    &\textstyle \bigwedge_{(i,j) \notin E^{\mathcal A}}
    \exists x \exists y(\pi_i(x)\land \pi_j(y) \land \neg E xy),
  \end{align*}
  where $\pi_i(x)$ is a shorthand for $\phi(x) \land
  \exists^{=i-1} z(z <x \land \phi(z))$, which binds $x$ to the $i$th
  element of the universe with property~$\phi$.
\end{notation}

\begin{notation}
  {$\op{induced}^{\mathrm{size}\le k}\{x \mid \phi(x)\} \in Q$}
  {$1+\operatorname{strong-qr}(\phi)+ \operatorname{arity}(\tau)$}
\item[Semantics] The set $I = \{ a \in |\mathcal A| \mid \mathcal A
  \models \phi(a)\}$ has size at most~$k$ and $\mathcal A[I] \in Q$.
\item[Expansion] $\bigvee_{s=0}^k \op{induced}^{\mathrm{size}=s}\{x \mid \phi(x)\} \in Q$  
\end{notation}

\subsection{Bounded Strong-Rank Description of Vertex Cover}

A \emph{vertex cover} of a graph $G = (V,E)$ is a subset $X \subseteq
V$ with $e \cap X \neq \emptyset$ for all $e \in E$. The problem
$\PLang[k]{vertex-set}$ asks whether a graph has a cover~$X$
with $|X| \le k$. 

\begin{theorem}[\cite{BannachST2015,ChenFH2017}]\label{thm-vc}
  $\PLang{vertex-cover} \in \Para\Class{AC}^0$.
\end{theorem}
\begin{proof}
  We describe the problem using a family $(\phi_k)_{k \in \mathbb N}$
  of constant strong quantifier rank that expresses the well-known
  Buss kernelization ``using logic'': 
  Let $\op{high}(x) = \exists^{\ge k+1} y\penalty50(Exy)$ expresses
  that $x$ is a \emph{high-degree vertex.} Buss observed that all
  high-degree vertices must be part of a 
  vertex cover of size at most~$k$. Thus, $h \le k$ must hold for the
  unique~$h$ with $\exists^{=h} x(\op{high}(x))$. A remaining vertex is
  \emph{interesting} if it is connected to at least one
  non-high-degree vertex: $\op{interesting}(x) = \neg \op{high}(x) \land \exists y(E xy \land \neg
  \op{high}(y))$. If there are more than $(k-h)(k+1) \le k^2 + k$
  interesting vertices, there cannot be a vertex cover -- and if there
  are less, the graph induced on the interesting vertices must have a
  vertex cover of size $k-h$. In symbols: $
    \phi_k = \textstyle \bigvee_{h=0}^k \bigl(
    \exists^{=h} x(\op{high}(x)) \land \op{induced}^{\mathrm{size}\le k^2 +
    k}\{ x \mid \op{interesting}(x)\} \in Q_{k-h}\bigr)$ for $Q_s = \{
  \mathcal G \mid  \mathcal G$~has a vertex cover of size~$s\}$.
\end{proof}

\subsection{Bounded Strong-Rank Description of Hitting Set}

Hitting sets generalize the notion of vertex covers to \emph{hypergraphs,}
which are pairs $(V, E)$ where the members of~$E$ are
called \emph{hyperedges} and we have $e \subseteq V$ for all $e\in
E$. Hitting sets are still sets $X \subseteq V$
with $e \cap X \neq \emptyset$ for all $e \in E$. The problem
$\PLang[k,d]{hitting-set}$ asks whether a hypergraph with $\max_{e\in
  E}|e| \le d$ has a hitting set $X$ with $|X| \le k$. Note that
$\PLang{vertex-cover}$ is exactly this problem restricted to $d=2$.

\begin{theorem}[\cite{BannachT2018}]\label{thm-hitting}
  $\PLang[k,d]{hitting-set} \in \Para\Class{AC}^0$.
\end{theorem}

Before we prove the theorem, let us fix how we model hypergraphs
$(V,E)$ as logical structures:  We use a $\tau_{\mathrm{hyper}}$-structure $\mathcal H$ 
for $\tau_{\mathrm{hyper}} =  (\op{vertex}^1, \op{hyperedge}^1, \op{in}^2)$.
Let $\left| \mathcal H \right| = V \cup E$, set $\op{vertex}^{\mathcal H} = V$, set
$\op{hyperedge}^{\mathcal H} = E$, and set $\op{in}^{\mathcal H} =
\{(v,e) \mid v \in e \in E\}$. The other way round, given a
$\tau_{\mathrm{hyper}}$-struc\-ture~$\mathcal H$, we consider it as the
following hypergraph $H(\mathcal H) = (V,E)$:
$V = \op{vertex}^{\mathcal H}$ and writing $\operatorname{set}(e)$ for
$\{v \mid (v,e) \in \op{in}^{\mathcal H}\}$ we set $E = \{
\operatorname{set}(e) \mid e \in \op{hyperedge}^{\mathcal H}\}$.

Note that we allow the universe of $\mathcal H$ to contain
elements~$e$ that are neither vertices nor hyperedge-representing
elements, but their $\operatorname{set}(e)$ do not contribute
to~$E$. We also allow that two different elements $e,e' \in |\mathcal
H|$ represent the same set $\operatorname{set}(e) =
\operatorname{set}(e')$. This can be problematic in a
kernelization: When we identify a kernel set~$E'$ of hyperedges, there
could still be a large (non-parameter-dependent) number of elements in
the universe that represent these hyperedges -- meaning that these
elements do not form a kernel themselves. Fortunately, this can be
fixed: We can easily check whether two elements represent the same set
using $\forall x (\op{in} xe \leftrightarrow \op{in} xe')$ and then
always consider only the first representing element with respect to the
ordering~$<$ of the universe. For this reason, we will assume in the
following that for any subset $s \subseteq V$ there is at most one $e
\in |\mathcal H|$ with $s = \operatorname{set}(e)$.

Let $d(H)$ be the maximum size of any hyperedge in~$H$ and let
$d(\mathcal H) = d(H(\mathcal H))$. 

A \emph{hitting set} for a hypergraph $(V,E)$ is a set $X \subseteq V$
with $e \cap X \neq \emptyset$ for all $e \in E$. The problem
$\PLang[k,d]{hitting-set}$ 
is the set of all pairs $(\mathcal H, \operatorname{num}(k,d))$ such
that $H(\mathcal H)$ is a hypergraph with $d(\mathcal H) \le d$ and for
which there is a hitting set of size at most~$k$.

\begin{proof}
  The idea behind the proof is a (very strong) generalization of the
  Buss kernel argument from the proof of Theorem~\ref{thm-vc}. As in
  that proof, we will present a family $(\phi_{k,d})_{k,d \in \mathbb
    N}$ of bounded strong quantifier rank that describes
  $\PLang[k,d]{hitting-set}$. First, there are two simple preliminaries:
  Testing whether $d(\mathcal H) \le d$
  holds is easy to achieve using $\forall e(\op{hyperedge} e\to
  \exists^{\le d}v(\op{in} v e))$, so let us
  assume that this is the case and let us write $H = (V,E)$ for
  $H(\mathcal H)$. Furthermore, let us write $\op{subset} e f$ for
  $\forall x(\op{in} xe \to \op{in} xf)$, which indicates that
  $\operatorname{set}(e) \subseteq \operatorname{set}(f)$.

  \subparagraph*{Representing Subsets of Hyperedges.}
  Recall that the core idea of the kernelization of the vertex cover
  problem is that a ``high-degree vertex'' must be
  part of a vertex cover. Rephrased in the language of 
  hypergraphs, a graph is a hypergraph $H$ with $d(H) = 2$, a vertex
  cover is a hitting set, and making a high-degree vertex~$v$ part of a
  hitting set is (in essence) the same as removing all edges
  containing~$v$ and then adding the singleton hyperedge $\{v\}$,
  which can clearly only be hit by making $v$ part of the hitting set.

  In the general case, we will also remove hyperedges from the
  hypergraph and replace them by smaller hyperedges (though, no
  longer, by singletons) and we will do so repeatedly. The problem is
  that adding hyperedges is difficult in our encoding since this means
  that we would have to add elements to the universe of the logical 
  structure that represent the new hyperedges. Although these problems
  can be circumvented by complex syntactic trickery, we feel it is
  cleaner to do the following at this point: We reduce the
  original hitting set problem to a new version, where the universe
  already contains all the necessary elements for representing the
  hyperedges we might wish to add later on.

  In detail, we define a subset $\PLang[k,d]{hitting-set}' \subseteq
  \PLang[k,d]{hitting-set}$ as follows: It contains only those
  $(\mathcal H,\operatorname{num}(k,d))$ such that for every $e \in
  \op{hyperedge}^{\mathcal H}$ and every subset $s \subseteq
  \operatorname{set}(e)$ there is an $e' \in \left|\mathcal H\right|$
  with $s = \operatorname{set}(e')$. In other words, for every
  subset~$s$ of any hyperedge there must already be an element $e$
  ``in store'' in the universe that represents it.

  We can reduce $\PLang[k,d]{hitting-set}$ to
  $\PLang[k,d]{hitting-set}'$ by adding for an input~$\mathcal H$, if
  necessary, elements to the universe that represent all these
  subsets. We are helped by the fact that we have an upper bound $d$
  on the size of the hyperedges, which means that the maximum blowup
  of the universe in this reduction is by the parameter-dependent
  value of~$2^d$. However, we have not yet defined which notion of
  \emph{reductions} between parameterized problems we wish to use and
  there are many definitions in the literature. Since
  $\Para\Class{AC}^0$ is severely restricted computation-wise, we must
  use a weak one.
  
  We postpone this question until after the proof, where we
  present a suitable definition for reductions
  (Definition~\ref{def-red}) such that all considered classes are
  closed under them and then show in Lemma~\ref{lem-hit} that 
  $\PLang[k,d]{hitting-set}$ reduces to $\PLang[k,d]{hitting-set}'$.
  Thus, in the following, we may assume that for every hyperedge in
  the input structure for all subsets of this hyperedge we already
  have an element in the universe representing this subset.

  \subparagraph*{Finding Sunflowers.}
  We first show a way of kernelizing the hitting set problem, due to
  Chen et~al.~\cite{ChenFH2017}, that ``almost works.'' 
  The core idea is to detect and collapse \emph{sunflowers} in the
  input hypergraph~\cite{ErdosR60}. A \emph{sunflower of size $k+1$ with
    core~$c$} is a set $\{p_1,\dots,p_{k+1}\} \subseteq E$ of distinct
  hyperedges, called \emph{petals}, such that for all $i \neq j$ we
  have $p_i \cap p_j = c$. In other words, all petals contain the core
  but are otherwise pairwise distinct. For convenience, we also assume
  that all petals are proper supersets of the core. The important
  observation is that if a sunflower of size $k+1$ has a hitting set of size~$k$,
  then the core must also be hit -- and when the core is hit, all
  petals are hit. This means that we can just replace a sunflower by
  its core when we are looking for size-$k$ hitting sets. 

  The following formula tests whether $\operatorname{set}(c)$ is the
  core of a sunflower of size $k+1$:
  \begin{align}
    \op{core} c =  \exists \dot p_1^1 \cdots \exists \dot p_{k+1}^d&\textstyle \bigwedge_{i\neq j} \bigwedge_{r,s \in \{1,\dots,d\}}
      \dot p_i^r \neq \dot p_j^s \land {} \label{eq-dist-sunflower} \\
    &\textstyle \bigwedge_{i=1}^{k+1} \exists e \bigl(\op{hyperedge} e
      \land 
      \op{subset} ce \land {}\notag \\
    & \phantom{\smash{\textstyle{\bigwedge_{i=1}^{k+1} \exists e
      \bigl(}}}  \{\dot p_i^1,\dots,\dot p_i^d\} = \{ v\mid \op{in} 
      v e \land \neg \op{in} vc\} \bigr). \label{eq-equal-sunflower}
  \end{align}
  Here, \eqref{eq-dist-sunflower} guarantees that the petals are
  pairwise disjoint outside the core and \eqref{eq-equal-sunflower}
  checks that the petals are supersets of $c$ and when we add $p_i^1$
  to $p_i^d$ (which are not necessarily disjoint) to $c$, we get a
  present hyperedge.

  The ``collapsing'' of sunflowers to their cores can now be done as
  follows: We define a formula with $e$ as a free variable that is
  true when $\operatorname{set}(e)$ is a core or when $\operatorname{set}(e)$ is not a superset of any core
  (otherwise, we need not include $\operatorname{set}(e)$ since we include the core of a
  sunflower that contains it, instead):
  \begin{align}
    \op{core} e \lor (\op{hyperedge} e \land \neg \exists c(\op{core}
    c \land \op{subset} ce)). \label{eq-kernel1}
  \end{align}
  The importance of the above formula lies in the following fact: The
  number of hyperedges for which the second part of the formula is
  true (that is, which are not supersets of a core of a sunflower of
  size $k+1$), is \emph{bounded by a function in $k$ and~$d$.} This is
  due to the famous Sunflower Lemma~\cite{ErdosR60} which states that
  if a hypergraph has more than $k^{d}d!$ hyperedges, it contains a
  sunflower of size~$k+1$ (which has a core).

  This means that if $\op{core} e$ were to hold for just a few
  hyperedges, \eqref{eq-kernel1} would describe a kernel
  for the hitting set problem and we would be done: Just as in the
  proof of Theorem~\ref{thm-vc}, we could use the $\op{induced}$
  notation to solve the hitting set problem on the vertices and
  hyperedges for which \eqref{eq-kernel1} holds.
  Unfortunately, it is possible to construct hypergraphs such that
  $\op{core} e$ still holds for a very large number of
  hyperedges.

  However, we know that a core always has \emph{a smaller size than
    any petal in its sunflower}. In particular, all cores have maximum
  size $d-1$. Thus, if we ``view $\op{core}$ as our new 
  $\op{hyperedge}$ predicate,'' we get ``cores of cores'':
  \begin{align*}
    \op{core}^2 c =  \exists \dot p_1^1 \cdots \exists \dot p_{k+1}^d& \textstyle \bigwedge_{i\neq j} \bigwedge_{r,s \in \{1,\dots,d\}}
      \dot p_i^r \neq \dot p_j^s \land {} \\
    &\textstyle \bigwedge_{i=1}^{k+1} \exists e' \bigl(\op{core} e'
      \land 
      \op{subset} ce' \land {}\notag \\
    & \phantom{\smash{\textstyle{\bigwedge_{i=1}^{k+1} \exists e \bigl(}}}  \{\dot p_i^1,\dots,\dot p_i^d\} = \{ v\mid \op{in}
      v e' \land \neg \op{in} vc\} \bigr).
  \end{align*}
  Note that $\operatorname{strong-qr}(\op{core}^2 c) =
  \operatorname{strong-qr}(\op{core} c) + 1 = 2$ since we had to add a
  new strong quantifier ($\exists e'$) whose scope contains $\op{core}
  e'$, which adds its own strong quantifier ($\exists e$).
  
  By the same argument as earlier, we get that the number of $e$ for
  which the following formula holds equals the number of cores of
  cores plus something that only depends on the parameters~$k$
  and~$d$: 
  \begin{align*}
    \op{core}^2 e &\lor (\op{core} e \land \neg \exists c(\op{core}^2
                    c \land \op{subset} ce)) \\
    &\lor (\op{hyperedge} e \land \neg \exists c(\op{core}
    c \land \op{subset} ce)).
  \end{align*}
  Still, the number of cores of cores can be large, but they all have
  size at most~$d-2$. Repeating the argument a further $d-2$
  times, we finally get the predicate $\op{kernel} e$:
  \begin{align}
    \op{core}^{d} e \lor \textstyle \bigvee_{i=1}^d
    (\op{core}^{i-1} e \land \neg \exists c(\op{core}^i c \land
    \op{subset} ce)), \label{eq-kernel}
  \end{align}
  where $\op{core}^0$ is of course $\op{hyperedge}$ and $\op{core}^d
  e$ can only be true for the (sole) $e$ representing the
  empty set (in which case, there is not hitting set). 

  Unfortunately, the strong quantifier rank of $\op{core}^d$ is $d$
  since the definition of $\op{core}^i$ in terms of $\op{core}^{i-1}$
  always adds one strong quantifier nesting (through a new $\exists
  e^{\prime \dots\prime}$). Thus, \eqref{eq-kernel} also has a strong
  quantifier rank of~$d$ while we need $O(1)$.

  \subparagraph*{Finding Pseudo-Sunflowers.}
  At this point, we need a way of describing cores of cores of cores
  and so on using a bounded strong quantifier rank. The idea how
  this can be done was presented in~\cite{BannachT2018}, where the
  notions of \emph{pseudo-cores} and \emph{pseudo-sunflowers} are
  introduced. The definitions are somewhat technical, see below, but
  the interesting fact about these definitions is that they can be
  expressed very nicely in a way similar to \eqref{eq-dist-sunflower}
  and \eqref{eq-equal-sunflower}.

  For a level $L$ and a number~$k$, let $T_L^k$ denote the rooted
  tree in which all leaves are at the same depth~$L$ and all inner nodes
  have exactly $k+1$ children. The root of $T_L^k$ will always be called
  $r$ in the following. Thus, $T_1^k$ is just a star consisting of $r$
  and its $k+1$ children, while in $T_2^k$ each of the $k+1$ children of
  $r$ has $k+1$ new children, leading to $(k+1)^2$ leaves in total. For
  each $l\in \operatorname{leaves}(T_L^k) = \{l \mid \text{$l$ is a leaf
    of $T_L^k$}\}$ there is a unique path $(l^0,l^1,\dots,l^L)$ from
  $l^0 = r$ to $l^L = l$.
  
  \begin{definition}[Pseudo-Sunflowers and Pseudo-Cores, \cite{BannachT2018}]\label{def-pseudo-core}
    Let $H = (V,E)$ be a hypergraph and let $L$ and $k$ be fixed. A set
    $c \subseteq V$ is called a \emph{$k$-pseudo-core of level~$L$ in
      $H$} if there exists a mapping $S \colon
    \operatorname{leaves}(T_L^k) \times \{0,1,\dots,L\} \to \{ e\mid
    e\subseteq V\}$, called a
    \emph{$T_L^k$-pseudo-sunflower for $H$ with pseudo-core~$c$}, such that for all $l, m \in
    \operatorname{leaves}(T_L^k)$ with $l\neq m$ we have:  
    \begin{enumerate}
    \item $S(l,0) = c$.
    \item $S(l,0) \cup S(l,1) \cup \dots \cup S(l,L) \in E$.
    \item $S(l,i) \cap S(l,j) = \emptyset$ for $0\le i < j \le L$, but $S(l,i)
      \neq \emptyset$ for $i \in \{1,\dots,L\}$.
    \item Let $z \in \{1,\dots, L\}$ be the smallest number such that
      $l^z \neq m^z$, that is, $z$ is the depth where the path from $r$
      to $l$ and the path from $r$ to $m$ diverge for the first
      time. Then $S(l,z) \cap S(m,z) =\emptyset$ must hold.
    \end{enumerate}
  \end{definition}

  This definition translates almost directly into a formula $\op{pseudocore}^L
  c$, which starts with 
  a block of weak existential quantifiers, one for each element of
  $\operatorname{leaves}(T_L^k) \times \{1,\dots,L\} \times \{1,\dots,d\}$: 
  \begin{align}
    & (\exists \dot x^j_{l,i})_{l \in
    \operatorname{leaves}(T_L^k),i \in
      \{1,\dots,L\},j\in\{1,\dots,d\}}\notag\\
    &\quad  \textstyle \bigwedge_{l,m\in
      \operatorname{leaves}(T_L^k),l\neq m,z\text{ as in the
      definition}} \smash{\Bigl(}\notag\\
    & \qquad 
    \exists e\bigl(\op{hyperedge} e\land \op{subset} ce \land
      \{\dot x^1_{l,1},\dots,\dot x_{l,L}^d\} = \{ v\mid \op{in}
      v e \land \neg \op{in} vc\} \bigr) \land
                                                              {} \label{eq-pseudo-one} \\ 
    & \qquad \textstyle \bigwedge_{i\neq j} \bigwedge_{p,q \in \{1,\dots,d\}}
      \dot x_{l,i}^p \neq \dot x_{l,j}^q \land {} \label{eq-pseudo-two}\\
    & \qquad \textstyle \bigwedge_{p,q \in
      \{1,\dots,d\}} \dot x_{l,z}^p \neq \dot x_{m,z}^q \smash{\Bigr).}\label{eq-pseudo-three}
  \end{align}
  Here, \eqref{eq-pseudo-one} ensures, similarly to
  \eqref{eq-equal-sunflower} for normal sunflowers, that $S(l,0) \cup S(l,1) \cup \dots \cup S(l,L)$ is a
  hyperedge, item~2 of the definition. The inequalities
  \eqref{eq-pseudo-two} ensure that item~3 of the definition holds,
  while \eqref{eq-pseudo-three} ensures item~4.

  The important observation is that $\op{pseudocore}^L$ has a strong
  quantifier rank that is \emph{independent} of~$L$. Since, as shown
  in~\cite{BannachT2018}, we can use $\op{pseudocore}^L$ as a
  replacement for $\op{core}^L$ in \eqref{eq-kernel}, we get that the
  hitting set problem can be described by a family of formulas of
  constant strong quantifier rank. 
\end{proof}

In the proof we used reductions (from
$\PLang[k,d]{hitting-set}$ to $\PLang[k,d]{hitting-set}'$) although we
have not yet given a definition of a notion of reductions 
that is appropriate for the context of the present paper. Clearly, we
need a notion of parameterized reductions that is very weak to ensure
that the smallest class we study, $\Para\Class{AC}^0$, is closed under
them. Such a reduction is used in the literature~\cite{BannachT2016},
boringly named $\Para\Class{AC}^0$-reduction, but both its definition
as well as the 
definition of other kinds of parameterized reductions found in the
literature do not fit well with our logical framework: The reductions
are defined in terms of machines or circuits that get as input a
string that explicitly or implicitly contains the parameter~$k$ and
output a new problem instance that once more explicitly or implicitly
contains a new parameter value~$k'$.

In contrast, in our setting the inputs and outputs must be logical
structures that we wish to define in terms of formulas. Furthermore,
``outputting a parameter value'' is difficult in our formal framework
since parameter values are not elements of the universe, but indices
of the formulas. All of these problems can be circumvented, see for
instance \cite[Definition~5.3]{ChenF2018}, but we believe it gives a
cleaner formalism to give a new ``purely 
logical'' definition of reductions between parameterized
problems. We will not prove this, but remark that the power of this
reduction is the same as that of $\Para\Class{AC}^0$-reductions.

\begin{definition}\label{def-red}
  Let $\tau$ and $\tau'$ be signatures and let $Q \subseteq
  \struc[\tau] \times \mathbb N$ and $Q' \subseteq \struc[\tau']
  \times \mathbb N$ be two problems. A \emph{bounded rank reduction}
  from $Q$ to~$Q'$, written $Q \le_{\mathrm{br}} Q'$,
  is a pair of computable families $(f_k)_{k \in\mathbb N}$ and
  $(\iota_{k,k'})_{k,k'\in\mathbb N}$ where 
  \begin{itemize}
  \item each $f_k$ is a first-order query from $\tau$-structures to
    $\tau'$-structures and 
  \item each $\iota_{k,k'}$ is a $\tau$-formula
  \end{itemize}
  such that
  \begin{enumerate}
  \item for each $(\mathcal A,k) \in \struc[\tau]\times \mathbb N$
    there is exactly one $k'\in\mathbb N$, denoted by $\iota_k(\mathcal
    A)$ in the following, such that $\mathcal A \models \iota_{k,k'}$,
  \item there is a computable mapping $\iota^* \colon \mathbb N
    \to\mathbb N$ such that for all $\mathcal A \in \struc[\tau]$ we
    have $\iota_k(\mathcal A)\le \iota^*(k)$,
  \item $(\mathcal A,k)\in Q$ if, and only if, $\bigl(f_k(\mathcal
    A),\iota_k(\mathcal A)\bigr)\in Q'$, and
  \item the quantifier rank of all $\iota_{k,k'}$ and of all formulas
    inside the~$f_k$ and of the widths of the~$f_k$ is bounded by a
    constant~$c$. 
  \end{enumerate}
\end{definition}

Let us briefly explain the ingredients of this definition: Each $f_k$
maps all $\tau$-structures $\mathcal A$ to
$\tau'$-structures~$\mathcal A'$. The fact that we have one function
for each parameter value allows us to make our 
mapping depend on the parameter. The job of the formulas $\iota_{k,k'}$ is
solely to ``compute'' the new parameter value~$k'$, based not only on
the original value~$k$, but also on~$\mathcal A$. If, as is the case
in many reductions, the new parameter value $k'$ just depends on $k$
(typically, it even \emph{is}~$k$), we can just set $\iota_{k,k'}$ to
a trivial tautology~$\top$ and all other $\iota_{k,k''}$ to the
contradiction~$\bot$.

In the definition, we referred to \emph{first-order queries,} which
are a standard way of defining a logical $\tau'$-structure in terms of
a $\tau$-structure. A detailed account can be found
in~\cite{Immerman1998}, but here is the basic idea: Suppose we wish to
map graphs ($(E^2)$-structures) to their underlying undirected graphs
($(U^2)$-structures, where $U$ represent the underlying symmetric edge
set). In this case, there is a simple formula $\phi_U(x,y)$ that 
tells us when $Uxy$ holds in the new structure: $Exy \lor Eyx$. 
More importantly, if we have a formula $\psi$ that internally uses
$Uxy$ to check whether there is an undirected edge in the mapped
graph, we can easily turn this into a formula $\psi[f]$, where we
replace all occurrences of $Uxy$ by $\phi_U(x,y)$, that 
gives the same answer as $\psi$ when fed the \emph{original} graph. In
other words, if a first-order query maps $\mathcal A$ to~$\mathcal A'$
and we wish to check whether $\mathcal A' \models \psi$ holds, we can just
as well check whether $\mathcal A \models \psi[f]$ holds.

The just-described example of a first-order query did not change the
universe, which is something we sometimes wish to do (indeed, the
whole point of the reduction between the two versions of the hitting
set problem was a change of the universe). This is achieved by
allowing the \emph{width}~$w$ of the query to be larger than~$1$. The
effect is that the universe $U$ gets replaced by $U^w$ and, now,
elements of this new universe can be described by tuples of variables
of length~$w$. We can also \emph{reduce} the size of the universe
using a formula $\phi_{\mathrm{universe}}(x_1,\dots,x_w)$ that is true
only for the tuples we wish to keep in the new structure's universe. 

\begin{lemma}\label{lem-closed}
  Let $Q \le_{\mathrm{br}} Q'$ via a bounded rank reduction given by
  $(f_k)_{k\in\mathbb N}$ and $(\iota_{k,k'})_{k,k'\in\mathbb N}$. Let
  $(\phi'_k)_{k\in \mathbb N}$ describe~$Q'$. Then there is a family
  $(\phi_k)_{k\in\mathbb N}$ that describes~$Q$ with
  \begin{enumerate}
  \item $\max_k \operatorname{qr}(\phi_k) = \max_k
    \operatorname{qr}(\phi'_k) + O(1)$ and
  \item $\max_k \left|\operatorname{bound}(\phi_k)\right| = \max_k
    \left|\operatorname{bound}(\phi'_k)\right| + O(1)$.    
  \end{enumerate}
\end{lemma}

In particular, $\Para\Class{AC}^0$ and
$\Para\Class{AC}^{0\uparrow}$ are closed under bounded rank
reductions.

\begin{proof}
  Set $\phi_k$ to $\bigwedge_{k'=1}^{\iota^*(k)} (\iota_{k,k'} \to
  \phi'_{k'}[f_k])$. By definition, we have $\mathcal A\models \phi_k$ 
  if, and only if, $f_k(\mathcal A) \models \phi'_{k'}$ for the
  unique $k'$ with $\mathcal A \models \iota_{k,k'}$. If we can
  argue that the substitutions do not increase the quantifier rank or
  number of variables by more than a constant, we get the claim.

  Unfortunately, simple substitutions fail to
  preserve the quantifier rank in a single case: When a formula $\phi'_{k'}$
  contains a large number of nested applications of the successor
  function. Suppose, for instance, $\phi'_{k'}$ is something like $\exists x
  \exists y(\op{succ}^{1000} x=y)$. While this formula has quantifier
  rank~$2$ and uses only two variables, a simple substitution of each
  occurrence of the one thousand $\op{succ}$ operators in~$\phi'_{k'}$ by
  any nontrivial formula in~$f_k$ that describes the successor function will
  yield a quantifier rank of at least~$1000$.

  The trick is to use color coding once more: We can
  easily modify any formula so that all occurrences of the successor
  function are of the form $x = \op{succ}^i 0$ for some number~$i$. This
  means that we ``only'' need a way of identifying the $i$th element
  of the new universe using a bounded quantifier rank. However,
  assuming for simplicity a width of~$1$ and assuming that
  $\phi_{\mathrm{universe}}(x)$ and $\phi_<(x,y)$ describe how $f_k$
  restricts the universe and possibly reorders it, respectively, the formula
  $\phi_{\mathrm{universe}}(x) \land \exists^{=i-1}y(\phi_{\mathrm{universe}}(y) \land
  \phi_<(y,x))$ is true exactly for the $i$th element of the universe
  -- and we saw already that we can express the $\exists^{=i-1} y$
  quantifier using a constant quantifier rank that is independent
  of~$i$. 
\end{proof}

\begin{example}
  We have $\PLang[k,\delta]{dominating-set} \le_{\mathrm{br}}
  \PLang[k,d]{hitting-set}$ where the first problem is parameterized
  by both the size~$k$ of the sought dominating set and a
  bound~$\delta$ on the maximum vertex 
  degree. For each parameter $(k,\delta)$, the first-order query $f_{k,\delta}$
  maps the input graph to the hypergraph where there is a hyperedge
  for the closed neighborhood of each vertex. This is achieved through
  $\phi_{\op{vertex}}(x) = \top$, $\phi_{\op{hyperedge}}(x) = \top$,
  and $\phi_{\op{in}}(x,y) = ((x=y) \lor Exy)$. The new parameter is
  $\delta+1$, which is achieved by
  $\iota_{\operatorname{num}(k,\delta),\operatorname{num}(k,\delta+1)}
  = \top$ and $\iota_{x,x'} = \bot$ otherwise. Observe that $\iota^*$
  is clearly computable. By Lemma~\ref{lem-closed} and since
  $\PLang[k,d]{hitting-set} \in \Para\Class{AC}^0$, we also have
  $\PLang[k,\delta]{dominating-set} \in \Para\Class{AC}^0$.
\end{example}

\begin{lemma}\label{lem-hit}
  $\PLang[k,d]{hitting-set} \le_{\mathrm{br}} \PLang[k,d]{hitting-set}'$.
\end{lemma}

\begin{proof}
  In the reduction, we do not change the parameter, so the
  $\iota_{\operatorname{num}(k,d),\operatorname{num}(k,d)} = \top$ and
  $\iota_{x,x'} = \bot$ otherwise. For the first-order queries, we wish to map a hypergraph 
  $\mathcal H$ to a new version $\mathcal H'$ in which for every
  subset of a hyperedge there is already an element in the universe
  representing this subset. This means that the size of the universe
  can increase from $\left|\mathcal H\right|$ to at most $2^d
  \left|\mathcal H\right|$. (If there is a hyperedge of size larger
  than $d$ in the input, we can yield a trivial ``no'' instance as
  output.)
  We use a first-order query of width~$2$, meaning that the universe
  size gets enlarged from $\left|\mathcal H\right|$ to $\left|\mathcal
    H\right|^2$. This will be larger than $2^d \left|\mathcal
    H\right|$ for all sufficiently large universes. Since, with
  respect to $f_{\operatorname{num}(k,d)}$ the number $2^d$ is a constant, we can apply
  Lemma~\ref{lemma-eventually} to take care of those inputs whose
  universes are smaller than $2^d$ and directly map them to the
  correct instances. For the large instances, we now have a universe
  that is ``large enough'' to contain an element for each subset of a
  hyperedge and it is not difficult (but technical) to use the bit
  predicate to define the correct predicates $\op{hyperedge}$,
  $\op{vertex}$, and $\op{in}$ in terms of the original structure. 
\end{proof}

\subsection{Bounded Strong-Rank Description of \\Model Checking for First-Order Logic}

An important result by Flum and Grohe~\cite{FlumG2003} states that
the model checking problem for first-order logic lies in $\Class{FPT}$
on structures whose Gaifman graph has bounded degree. Once more, this
result can now be obtained ``syntactically.'' For simplicity, we only consider graphs and let 
$\PLang[\psi,\delta]{mc}(\Class{FO}) = \bigl\{ (\mathcal
G,\operatorname{num}(\psi,\delta)) \bigm|
\mathcal G \in \struc[(E^2)], \psi \in \Class{FO}, \mathcal G \models \psi, \operatorname{max-degree}(\mathcal G) \le
\delta\bigr\}$.

\begin{theorem}[\cite{BannachST2015,FlumG2003}]\label{thm-mc}
  $\PLang[\psi,\delta]{mc}(\Class{FO}) \in \Para\Class{AC}^{0\uparrow}$.
\end{theorem}

\begin{proof}
  We present a family $(\phi_{\psi,\delta})_{\psi \in
    \Class{FO},\delta \in \mathbb N}$ with a bound on the number of 
  strong variables that describes $\PLang[\psi,\delta]{mc}(\Class{FO})$.
  Fix $\psi$ and $\delta$. Recall that we fixed the signature
  for~$\psi$ to just $\tau = (E^2)$ for simplicity and, thus,
  $\tau$-structures are just graphs~$\mathcal G$. In particular, there
  are \emph{no} arithmetic predicates available to~$\psi$ (one could,
  of course, also consider them, but then the Gaifman graph would
  always be a clique and the claim of the theorem would be boring). In
  contrast, the $\phi_{\psi,\delta}$ are normal $\Class{FO}[+,\times]$
  formulas and \emph{they} have access to arithmetics. 

  For a graph $\mathcal G$ let us write $\bar{\mathcal G}$ for the
  underlying undirected graph and let us write $\bar Exy$ as a
  shorthand for $Exy \lor Eyx$. The first thing we check is that the
  maximum degree of the input graph $\bar{\mathcal G}$ is, indeed,
  $\delta$. This is rather easy: $\forall x \exists^{\le \delta}
  y(\bar Exy)$. 

  For the hard part of determining whether $\mathcal G\models \psi$,
  let $\bar d(a,b)$ denote the distance of two vertices in~$\bar{\mathcal
  G}$ and let $N_r(a) = \bigl\{b \in |\mathcal G| \bigm| \bar d(a,b) \le r\bigr\}$
  be the ball around $a$ of radius~$r$ in~$\bar{\mathcal G}$. Let
  $\mathcal G[N_r(a)]$ denote the subgraph of $\mathcal G$ induced on 
  $N_r(a)$. By Gaifman's Theorem~\cite{Gaifman1982} we can rewrite
  $\psi$ as a Boolean combination of formulas of the following form:
  \begin{align} 
   \exists x_1 \cdots \exists x_k \Bigl(
    &\textstyle\bigwedge_{i\neq j}
      \gamma_{\bar d(x_i,x_j)>2r}  \land {} \label{eq-d}\\
    &\textstyle\bigwedge_i \rho(x_i)\Bigr) \label{eq-local}
  \end{align}
  where $\gamma_{\bar d(x_i,x_j)>2r}$ expresses, of course, that $\bar
  d(x_i,x_j) > 2r$ should hold and $\rho$ is \emph{$r$-local,} meaning
  that for all $a \in |\mathcal G|$ we have $\mathcal G \models \rho(a) \iff
  \mathcal G[N_r(a)] \models \rho(a)$ (the minimum number for which is
  the case is called the \emph{locality rank} of~$\rho$).

  We now wish to express the above formula using only a constant
  number of strong variables. The problem is, of course, that the
  $x_i$ are not (yet) weak since they are used many times. We fix this
  in two steps. First, let us tackle \eqref{eq-d}: Clearly, the
  $x_i$ will have a pairwise distance of at least $2r$, if the balls of
  radius~$r$ that surround them are pairwise disjoint. Now, because of
  the bounded degree of the graph, a ball of radius~$r$ can have
  maximum size $\delta^r$. This allows us to bind all members of each
  ball and testing disjointness is, of course, what weak variables are
  all about.

  In detail, let $\gamma_{\bar d(x,y) \le r}$ be the standard formula with
  two bound variables expressing that there is path from
  $x$ to~$y$ of length at most~$r$ in $\bar{\mathcal G}$. Then we can
  express \eqref{eq-d} as follows:
  \begin{align*}
    \exists \dot x_1 \cdots \exists \dot x_k 
    \,\exists \dot y_1^1 \cdots \exists \dot y_k^{\delta^r} \bigl(&\textstyle
      \bigwedge_{i\neq j} \bigwedge_{p,q \in \{1,\dots,\delta^r\}}
      \dot y_i^p \neq \dot y_j^q \land {}\\
    &\textstyle \bigwedge_{i=1}^k \exists x (x = \dot x_i
    \land \{\dot y_i^1,\dots, \dot y_i^{\delta^r}\} = \{ y \mid
    \gamma_{\bar d(x,y) \le r}\} )\bigr).
  \end{align*}
  In the formula, at the end we bind the variables $\dot
  y_i^1,\dots,\dot y_i^{\delta^r}$ exactly to the elements of the ball
  around $\dot x_i$ or radius~$r$; and in the first part we require
  that all these balls are pairwise disjoint. Note that we do
  \emph{not} require all $\dot y_i^p$ to be different: If the size of
  a ball is less than $\delta^r$, we must allow some $\dot y_i^p$ and
  $\dot y_i^q$ to be identical.

  In order to express \eqref{eq-local}, we just have to check for each
  $\dot x_i$ that the ball of radius $\delta^r$ around it is a model
  of $\rho(\dot x_i)$. Since the size of this ball is at most
  $\delta^r$, we can use the $\op{induced}$ notation. There is,
  however, a technical problem: We basically wish to check whether
  $\mathcal G[N_r(a)] \in \{ \mathcal H \mid \mathcal H \models
  \rho(a)\}$ holds for a given~$a$, but $\{ \mathcal H \mid \mathcal H \models
  \rho(a)\}$ obviously depends on~$a$ --
  which is not compatible with the $\op{induced}$
  notation. Fortunately, this problem can be fixed: For $i \in \mathbb
  N$ let $Q_i = \{ \mathcal H \mid i \le \|\mathcal H \|$, $a$~is the
  $i$th element of $|\mathcal H|$ with respect to~$<^{\mathcal H}$, 
  $\mathcal H \models \rho(a)\}$. If we know for some element $a \in
  |\mathcal G|$ that it is the $i$th element in $N_r(a)$, then our
  problematic test can be replaced by $\mathcal G[N_r(a)] \in
  Q_i$. Since testing whether $a$ is the $i$th element in $N_r(a)$ is
  possible using a formula like $\iota_i(a) = \exists^{=i-1} b(b<a
  \land \gamma_{\bar d(a,b) \le r})$, we get the following complete
  formula~$\phi_{\psi,\delta}$: 
  \begin{align*}
    \exists \dot x_1 \cdots \exists \dot x_k 
    \,\exists \dot y_1^1 \cdots \exists \dot y_k^{\delta^r}\smash{\Bigl(}& \textstyle
      \bigwedge_{i\neq j} \bigwedge_{p,q \in \{1,\dots,\delta^r\}}
      \dot y_i^p \neq \dot y_j^p \land {}\\
    &\textstyle \bigwedge_{i=1}^k \exists x \bigl(x = \dot x_i
    \land \{\dot y_i^1,\dots, \dot y_i^{\delta^r}\} = \{ y \mid
      \gamma_{\bar d(x,y) \le r}\} \land {} \\
    &\textstyle \phantom{\bigwedge_{i=1}^k \exists x \bigl(}
      \bigvee_{i=1}^{\delta^r} (\iota_i(x) \land
      \op{induced}^{\mathrm{size}\le \delta^r}\{ y \mid \gamma_{\bar
      d(x,y) \le r} \} \in Q_i)\bigr)\smash{\Bigr)}.
  \end{align*}
  This formula uses only a constant number of strong variables. Its
  strong quantifier rank would also be constant \emph{except} that the
  formula $\gamma_{\bar d(x,y) \le r}$ uses $r$ nested (strong)
  quantifiers (but only $2$ variables). This means that the strong
  quantifier rank of $\phi_{\psi,\delta}$ will be
  $O(\operatorname{locality-rank}(\psi))$. 
\end{proof}

\subsection{Bounded Strong-Rank Description of\\ Embedding Graphs of
  Constant Tree~Width~or~Constant Tree Depth}

For our final example, a graph $H = (V(H),E(H))$ \emph{embeds into} a graph $G =
(V(G),E(G))$ if there is an injective mapping $\iota \colon V(H) \to V(G)$
such that for all $(u,v) \in E(H)$ we have $(\iota(u),\iota(v))\in
E(G)$. We wish to show that the embedding problems for graphs of
bounded tree depth or bounded tree width lie in $\Para\Class{AC}^0$
and $\Para\Class{AC}^{0\uparrow}$, respectively:

\begin{theorem}[\cite{BannachST2015,ChenM2015}]\label{thm-emb}
  $\PLang{emb}_{\mathrm{td}{\le} c}
  \in \Para\Class{AC}^0$ and $\PLang{emb}_{\mathrm{tw}{\le} c}
  \in \Para\Class{AC}^{0\uparrow}$ for each~$c$.
\end{theorem}
Of course, we still need to review the underlying definitions: For a graph $H =
(V(H),E(H))$, a \emph{tree decomposition} of~$H$ is a tree $T =
(V(T), E(T))$ (a connected, acyclic, undirected graph) together with
a mapping $B$ that assigns a subset of $V(H)$ to each node in
$V(T)$. These subsets are called \emph{bags} and must have two
properties: First, for every edge $\{u,v\} \in E(H)$ there must be a 
node $n \in V(T)$ with $u,v \in B(n)$. Second, for each vertex $v
\in V(H)$ the set 
$\{n \in V(T) \mid v \in B(n)\}$ must be nonempty and connected in~$T$. Let
$\operatorname{width}(B) = \max_{n\in V(T)} |B(n)| - 1$. The
\emph{tree width $\operatorname{tw}(H)$ of~$H$} is the minimum width
of a tree decomposition for it.

We call $(T,B)$ a \emph{tree-depth decomposition} if $T$
can be rooted in such a way that if $u$ lies on the path from some
vertex $v$ to the root, then $B(u) \subsetneq B(v)$. The tree depth
$\operatorname{td}(H)$ is the minimum width of a tree-depth
decomposition $(T,B)$ of~$H$ plus~$1$. Note that this width is an
upper bound on $\operatorname{depth}(T)$, the depth of~$T$.

\begin{proof}
  We present a family $(\phi_{H,T,B})$ of $\tau$-formulas (where
  $\tau = (E^2, <^2, \op{succ}^1, \op{add}^3,\penalty 10 \op{mult}^3,\penalty 100 0^0)$ is the
  arithmetic signature of graphs) indexed by graphs $H$ together with
  any tree decomposition $(T,B)$ of~$H$ (without bounds on the depth or width)
  that describe the embedding problem. More precisely, we show the
  following: 

  \begin{claim*}
    There is a family $(\phi_{H,T,B})_{H \in
    \struc[\tau],\text{\normalfont$(T,B)$ is a tree decomposition of $H$}}$ such that:
    \begin{enumerate}
    \item $\mathcal G \models \phi_{H,T,B}$ if, and only if, $H$
      embeds into $\mathcal G$ (more precisely, into $(|\mathcal G|,
      E^{\mathcal G})$).
    \item $\operatorname{strong-qr}(\phi_{H,T,B}) = \operatorname{depth}(T)$.
    \item $\left|\operatorname{strong-bound}(\phi_{H,T,B})\right| =
      \operatorname{width}(B)+1$. 
    \end{enumerate}
  \end{claim*}

  \begin{proof}
    Before we present the formula, we define what we will call a
    \emph{consistent numbering} of the vertices of $H$. It is a 
    mapping $p \colon V(H) \to \{1,\dots,m\}$, where $m$ is the
    maximum bag size of the decomposition (so $m =
    \operatorname{width}(B)+1$). The number $p(v)$ for $v \in V(H)$
    can be thought as the ``position'' or ``index'' of~$v$ in all bags
    that contain it, that is, we require that for any bag $B(n) =
    \{b_1,\dots, b_{|B(n)|}\}$ the values $p(b_1),\dots,p(b_{|B(n)|})$
    are all different. (Phrased differently, $p$ restricted to any bag
    is injective.) Such a consistent numbering can be obtained as
    follows: First, assign the numbers $1$ to $|B(r)|$ to the elements
    of the root bag~$B(r)$. Now, consider a child~$c$ of the root~$r$
    in~$T$. The bag $B(c)$ may miss some of the elements of $B(r)$ and
    there may be some new elements. For each new element~$e$, let
    $p(e)$ be a different number from the set $\{1,\dots,m\}
    \setminus \{p(v) \mid v \in B(r) \cap B(c)\}$ and note that we will not run out of
    numbers. We assign numbers to all elements in the bags of the
    children of the root in this way and, then, we recursively use the
    same method for the children's children and so on. Note that, not
    only, we do not run out of numbers, but the consistency condition
    is also met: Once an element drops out of a bag, we will never see
    it again in a later bag and, hence, we cannot inadvertently assign
    a different number to it later on. 

    As a running example, we will use the graph~$H$ and the tree
    decomposition $(T,B)$ of it from Figure~\ref{fig:ex:decomp}. 

    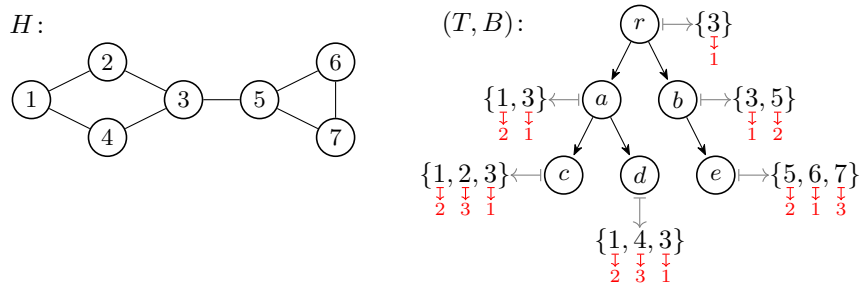
\begin{figure}[htpb]
      \centering
    \begin{tikzpicture}[graph]
      \node at (0,1) {$H\colon$};
      \graph [no placement, nodes=node, math nodes] {
        1 [x=0, y=0];
        2 [x=1, y=.5];
        3 [x=2, y=0];
        4 [x=1, y=-.5];
        5 [x=3, y=0];
        6 [x=4, y=.5];
        7 [x=4, y=-.5];
        1 -- {2,4} -- 3 -- 5 -- 6 -- 7 -- 5;
      };

      \scoped[xshift=7cm] {
        \node at (-1,1) {$(T,B)\colon$};

        \graph [no placement, nodes=node, math nodes] {
          r   [x=1, y=1];
          a [x=.5, y=0];
          b [x=1.5, y=0];
          c [x=0, y=-1];
          d [x=1, y=-1];
          e [x=2, y=-1];
          r -> {
            a -> {
              c,
              d
            },
            b -> e
          }
        };
      }

      \def\p#1{\tikz[overlay]{\node[inner sep=1pt,red,below] (x) at
          (-.55ex,-7pt) {\scriptsize#1};\draw [red,{|[sep,scale=.6]}-{To[scale=.6]}] (-.55ex,-1pt) -- (x);}}
      
      \foreach \n/\bag/\dir in {r/{3\p1}/180, a/{1\p2,3\p1}/0, c/{1\p2,2\p3,3\p1}/0, d/{1\p2,4\p3,3\p1}/90, b/{3\p1,5\p2}/180, e/{5\p2,6\p1,7\p3}/180} {
        \node (x) at ([shift={(\dir-180:7mm)}]\n.center)[anchor=\dir,inner sep=0pt] {$\{\bag\}$};
        \path ([yshift=-11pt]x.south);
        \draw [{|[sep]}-To,gray] (\n) -- (x.\dir);
      }
    \end{tikzpicture}
      \caption{An example graph~$H$ together with a tree decomposition
      for it, consisting of the tree~$T$ and the bag function~$B$
      indicated using the small gray mapping arrows. A consistent
      numbering $p$ is indicated in red.}
      \label{fig:ex:decomp}
    \end{figure}

    The consistent numbering indicated in Figure~\ref{fig:ex:decomp}
    is obtained by mapping the vertices in the root node's bag (just
    $3$ in the example) to the index~$1$, so $p(3) = 1$. For the child node~$a$, the bag
    $\{1,3\}$ contains a new vertex, namely~$1$, which gets the next
    free index, in this case $p(1) = 2$. In the same way, for the
    other child~$b$ of the root, the new vertex~$5$ also gets the
    index~$2$. For the leaves, in the bags of $c$ and~$d$ we just have
    an additional vertex, which gets the last free index and, thus,
    $p(2) = 3$ and $p(4)=3$. For $e$ with $B(e) = \{5,6,7\}$, we must
    \emph{reuse} a number for the first time: from $B(b)$ to $B(e)$,
    the number $3$ drops out of the bag and, thus, we can (even must)
    reuse its index (which was $1$) for one of the nodes $6$
    or~$7$. Let us set $p(6) = 1$ and $p(7) =3$.     
    
    Let us now define $\phi_{H,T,B}$. We may assume $V(H) =
    \{1,\dots,|V(H)|\}$. 
    The first step is to bind all vertices of~$H$ to weak variables
    using $\exists \dot x_1\cdots \exists \dot x_{|V(H)|}
    \bigwedge_{i\neq j} \dot x_i \neq \dot x_j \land \psi$, where
    $\psi$ must now express that the bound elements form an
    embedding. To achieve this, we build $\psi$ recursively, starting
    at the root $r$ of~$T$ and $\psi = \psi_r$.

    For our example, we  would have: 
    \begin{align*}
      \textstyle \phi_{H,T,B} = \exists \dot x_1\cdots \exists \dot x_7
      \bigwedge_{i\neq j} \dot x_i \neq \dot x_j \land \psi_r. 
    \end{align*}

    For any node $n \in
    V(T)$, let the elements of the set $B(n)$ be named $b_1$ to $b_s$
    (these are just temporary names that have nothing to do with the
    consisting numbering~$p$) and let the first~$t$ of them be
    new, that is, not present in the parent bag (for the root, $s=t$
    and all elements are new; if there are no new elements,
    $t=0$). The formula $\psi_n$ will now express 
    the following: First, it binds the new elements using strong
    variables that are made equal to the weak variables representing the
    elements in the input structure. This makes the new strong
    variables disjoint from one another and also from from all other
    (images of) vertices of~$H$. Second, we check that for all edges
    $\{x,y\} \in E(H)$ between elements $x$ and $y$ of $B(n)$, that their
    images (which we have been bound to strong variables) are also connected in
    the input structure. Third, we require that these properties also
    hold for all children of~$n$. In symbols, we set:
    \begin{align*}
      \psi_n = \exists v_{p(b_1)} \cdots \exists v_{p(b_t)}\bigl(
      &v_{p(b_1)} = \dot x_{b_1} \land \cdots 
      \land v_{p(b_t)} = \dot x_{b_t} \land {} \\
      &\textstyle \bigwedge_{x,y \in B(n), \{x,y\} \in E(H)} E v_{p(x)}
      v_{p(y)} \land {} \\
      &\textstyle \bigwedge_{c \in
      \operatorname{children}(n)} \psi_c \bigr).
    \end{align*}

    For our example, let us start with the root~$r$. Here, we have
    $B(r) = \{3\}$ and $p(3) =1$ and there are no edges between the
    vertices in the bag (there is just one vertex, after all). This
    yields: $\psi_n= \exists v_1 (v_1 = \dot x_3 \land \psi_a \land
    \psi_b)$.

    For the node~$a$, a new node ($1$) enters the bag $B(a)$ with
    index $2 = p(1)$, but there are no intra-bag edges, so $\psi_a =
    \exists v_2 (v_2 = \dot x_1 \land \psi_c \land \psi_d)$.

    For the node~$b$ the situation is very similar, but there is now
    an edge $\{3,5\}$ in~$H$. This means that we must check that the
    nodes~$v_1$, representing~$3$, and $v_2$, representing~$5$, are
    connected in the input structure. This yields $\psi_b=
    \exists v_2 (v_2 = \dot x_5 \land E v_1 v_2\land \psi_e)$.

    For the node~$c$, we only have one new node ($2$) with a new index
    ($3 = p(2)$), but now there are two intra-bag edges in $H$, namely
    $\{1,2\} \in E(H)$ and $\{2,3\} \in E(H)$. This yields: $\psi_c =
    \exists v_3 (v_3 = \dot x_2 \land E v_3 v_1 \land E v_3 v_2)$,
    where $E v_1 v_3$ checks whether for $\{2,3\} \in E(H)$ there is a
    corresponding edge in the input structure (recall that $p(2) = 3$
    and $p(3) = 1$) and $E v_3 v_2$ checks the same for $\{1,2\}$.

    In a similar way, we get $\psi_d = \exists v_3 (v_3 = \dot x_4
    \land E v_3 v_1 \land E v_3 v_2)$ and observe that the only
    difference is that $v_3$ is made equal to $\dot x_4$ instead of
    $\dot x_2$, the rest is the same.

    Finally, for the node~$e$, we bind two strong variables since
    there are two new vertices ($6$ and~$7$), but we reuse variable
    $v_1$ for $6$ since the vertex~$3$ that used to have index~$1$ has
    dropped out of the bag. We get $\psi_e = \exists v_1 \exists
    v_3(v_1 = \dot x_6 \land v_3 = \dot x_7 \land E v_1 v_2 \land E
    v_1 v_3 \land E v_2 v_3)$.

    Putting it all together, we get the following $\phi_{H,T,B}$,
    whose structure closely mirrors~$T$'s:
    \begin{align*}
      \exists \dot x_1\cdots \exists \dot x_7
      &\textstyle \bigwedge_{i\neq j} \dot x_i \neq \dot x_j \land {} \\
      & \exists v_1 (v_1 = \dot x_3 \land {} \\
      &\quad  \exists v_2 (v_2 = \dot x_1 \land {} \\
      &\qquad \exists v_3 (v_3 = \dot x_2 \land E v_3 v_1 \land E v_3 v_2)\land {} \\
      &\qquad \exists v_3 (v_3 = \dot x_4 \land E v_3 v_1 \land E v_3
        v_2) )\land {} \\
      &
      \exists v_2 (v_2 = \dot x_5 \land E v_1 v_2\land {} \\
      &\qquad
          \exists v_1 \exists v_3(v_1 = \dot x_6 \land v_3 = \dot x_7 \land E v_1 v_2 \land E v_1 v_3 \land E v_2 v_3)
        )
      ).
    \end{align*}
    It remains to argue that $\phi_{H,T,B}$ has the claimed
    properties. Clearly, by construction, the strong quantifier
    rank and number of strong bound variables are as claimed. The
    semantic correctness also follows easily from the construction: If
    the input structure is a model of the formula then, clearly, the
    assignments of the $\dot x_i$ to elements of the universe form an
    embedding since for every edge $\{u,v\} \in E(H)$ somewhere in the
    formula we test whether $E v_{p(u)} v_{p(v)}$ holds where
    $v_{p(u)}$ is equal to $\dot x_u$ and $v_{p(v)}$ to $\dot
    x_v$. The other way round, given a model of the formula, any
    assignment to the $\dot x_i$ that makes it true is an embedding
    since, first, we require that all $\dot x_i$ are different and we
    require $E v_{p(u)} v_{p(v)}$ for all $\{u,v\} \in E(H)$. This
    concludes the proof of the claim.
  \end{proof}

  With the claim established, we can now easily derive the statement
  of the theorem. To show $\PLang{emb}_{\operatorname{td}\le c}
  \in \Para\Class{AC}^0$, we must present a family
  $(\phi_H)_{H\in\struc[\tau],\operatorname{td}(H)\le c}$ that
  describes $\PLang{emb}_{\operatorname{td}\le c}$ and that has
  bounded quantifier rank. Clearly, we can just set $\phi_H$ to
  $\phi_{H,T,B}$ where $(T,B)$ is a tree-depth decomposition of $H$ of
  depth~$c$ (which must exist by the assumption that
  $\operatorname{td}(H) \le c$). The second item of the claim
  immediately tells us that all $\phi_H$ will have a strong quantifier
  rank of at most~$c$; and we can use the characterization of
  $\Para\Class{AC}^{0}$ from Fact~\ref{fact-ac0}.
  For the second statement,  $\PLang{emb}_{\operatorname{tw}\le c}
  \in \Para\Class{AC}^{0\uparrow}$, we use a different family
  $(\psi_H)_{H\in\struc[\tau],\operatorname{tw}(H)\le c}$, this time
  setting $\psi_H$ to 
  $\phi_{H,T,B}$ where $(T,B)$ is a tree decomposition of $H$ of
  width~$c$. Now the third item of the claim gives us the bound on the
  number of strong variables; and we can use the characterization of
  $\Para\Class{AC}^{0\uparrow}$ from Theorem~\ref{thm-ac0up}.
\end{proof}

\section{Conclusion}

In the present paper, we showed how the color coding technique can be 
turned into a powerful tool for parameterized descriptive complexity
theory. This tool allows us to show that important results from
parameterized complexity theory -- like the fact that the embedding
problem for graphs of bounded tree width lies in $\Class{FPT}$ --
follow just from the syntactic structure of the formulas that describe
the problem. 

In all our syntactic characterizations it was important that variables
or color predicates were not allowed to be within a universal
scope. The reason was that literals, disjunctions, conjunctions, and
existential quantifiers all have what we called the \emph{small witness
  property,} which universal quantifiers do \emph{not} have. However,
there are other quantifiers, from more powerful logics that we did not
explore, that also have the small witness property. An example are
operators that test whether there is a path of length at most~$k$ from
one vertex to another for some fixed~$k$: if such a path exists, its
vertices form a ``small witness.'' Weak variables may be used inside
these operators, leading to broader classes of problems that can be 
described by families of bounded strong quantifier rank. On
the other hand, we cannot add the full transitive closure operator
$\op{tc}$ (for which it is well-known that $\Class{FO}[\op{tc}] =
\Class{NL}$) and hope that Theorems 
\ref{thm-cc-intro} and~\ref{theorem-strong-character} still hold: If
this were the case, we should be able to turn a formula that uses two
colors $C_1$ and~$C_2$ to express that there are two vertex-disjoint
paths between two vertices into a $\Class{FO}[\op{tc}]$ formula --
thus proving the unlikely result that the $\Class{NP}$-hard disjoint
path problem is in~$\Class{NL}$.

Another line of inquiry into the descriptive complexity of
parameterized problems was already started in the repeatedly cited
paper by Chen et al.~\cite{ChenFH2017}: They give first syntactic
properties for families of formulas describing weighted model checking
problems
that imply membership in $\Para\Class{AC}^0$. We believe that it might
be possible to base an alternative notion of weak quantifiers on these 
syntactic properties. Ideally, we would like to prove a theorem
similar to Theorem~\ref{theorem-strong-character} in which there are
just more quantifiers that count as weak and, hence, even more
families have bounded strong quantifier rank. This would allow us to
prove for even more problems that they lie in $\Class{FPT}$ just
because of the syntactic structure of the natural formula families
that describe them.

\bibliography{2018-color-coding}

\end{document}